\documentclass[a4paper,onecolumn,11pt,noarxiv]{quantumarticle}
\usepackage[utf8]{inputenc}
\usepackage{csquotes}
\usepackage[english]{babel}
\usepackage[T1]{fontenc}

\usepackage{amsmath,amsthm,amsfonts,amssymb,dsfont,graphicx,mathtools,physics,thmtools,thm-restate}
\usepackage{xprintlen}

\usepackage{complexity}
\usepackage{enumitem}
\usepackage{fullpage}

\usepackage{comment}
\usepackage{algorithm}  
\usepackage{algpseudocode}
\algrenewcommand\algorithmicrequire{\textbf{Input:}}
\algrenewcommand\algorithmicensure{\textbf{Output:}}

\usepackage{xcolor}
\usepackage{hyperref}
\hypersetup{
  colorlinks=true,
  hypertexnames=false,
  linktocpage,
  colorlinks=true, 
  urlcolor=magenta!90!black,    
  linkcolor=blue!60!black, 
  citecolor=black!60 
}
\usepackage[capitalize, compress]{cleveref}

\crefname{problem}{Problem}{Problems}

\usepackage[
  backend=biber,
  citestyle=alphabetic,
  bibstyle=alphabetic,
  minalphanames=3,maxalphanames=4,
  maxnames=10,minnames=1,
  maxcitenames=3,mincitenames=1,
  eprint=true,
  date=year,
  doi=false,
  url=false,
  isbn=false,
  ]{biblatex}
\AtEveryBibitem{
  \clearfield{urldate}
  \clearfield{primaryclass}
}

\newclass{\stoqma}{StoqMA}
\newclass{\classp}{P}
\newclass{\bqp}{BQP}
\newclass{\postbqp}{postBQP}
\newclass{\posta}{postA}
\newclass{\postiqp}{postIQP}
\newclass{\classa}{A}
\newclass{\bpp}{BPP}
\newclass{\fbpp}{FBPP}
\newclass{\pp}{PP}
\newclass{\ceqp}{C_=P}
\newclass{\ph}{PH}
\newclass{\np}{NP}
\newclass{\gapp}{GapP}
\newclass{\approxclass}{Apx}
\newclass{\gapclass}{Gap}
\newclass{\sharpP}{\#P}
\newclass{\ma}{MA}

\newclass{\am}{AM}
\newclass{\coam}{coAM}
\newclass{\qma}{QMA}
\newclass{\qcma}{QCMA}
\newclass{\qcam}{QCAM}
\newclass{\Time}{TIME}
\newclass{\ti}{TI}
\newclass{\gi}{GI}
\newclass{\classx}{X}
\newclass{\mis}{mis}
\newclass{\pis}{pis}
\newclass{\MIS}{MIS}
\newclass{\PIS}{PlantedIndSp}
\newclass{\IS}{IndSp}
\newclass{\search}{search}
\newclass{\decision}{decision}
\newclass{\ips}{IP1S}
\newclass{\isop}{IP}
\newclass{\hog}{HOG}
\newclass{\quath}{QUATH}
\newclass{\bog}{BOG}
\newclass{\xeb}{XEB}
\newclass{\llqsv}{LLQSV}
\newclass{\llha}{LLHA}
\newclass{\Exp}{Exp}
\newclass{\xhog}{XHOG}
\newclass{\xquath}{XQUATH}
\newclass{\maxcut}{MAXCUT}
\newclass{\sat}{SAT}
\newclass{\maxtwosat}{MAX2SAT}
\newclass{\twosat}{2SAT}
\newclass{\threesat}{3SAT}
\newclass{\sharpsat}{\#SAT}

\newtheorem{theorem}{Theorem}
\newtheorem{proposition}[theorem]{Proposition}
\newtheorem{conjecture}[theorem]{Conjecture}
\newtheorem{definition}[theorem]{Definition}
\newtheorem{lemma}[theorem]{Lemma}
\newtheorem{problem}[theorem]{Problem}

\newtheorem{task}[theorem]{Task}

\newtheorem{corollary}[theorem]{Corollary}

\DeclareMathOperator{\supp}{supp}

\DeclareMathOperator{\Var}{Var}
\DeclareMathOperator{\spn}{span}

\newcommand{\gl}{\mathrm{GL}}

\newcommand{\mc}{\mathcal}
\newcommand{\mb}{\mathbb}

\newcommand{\id}{\mathbbm{1}}

\DeclareMathOperator*{\Eb}{\mb E}

\newcommand{\bin}{\{0,1\}}

\newcommand{\proj}[1]{\ket{#1}\bra{#1}}

\DeclareMathOperator{\gap}{gap}
\DeclareMathOperator{\ngap}{ngap}
\DeclareMathOperator{\negl}{negl}

\def\FF{\mathbb{F}}
\def\NN{\mathbb{N}}

\DeclareMathOperator{\Hess}{H}
\DeclareMathOperator{\Quad}{Q}
\DeclareMathOperator{\TF}{T}

\newcommand{\SR}{\operatorname{SR}} 
\newcommand{\symSR}{\operatorname{symSR}} 
\newcommand{\Str}{\operatorname{Str}} 

\newtheorem{remark}[theorem]{Remark}

\newtheorem{fact}[theorem]{Fact}

\begin{document}

\title{Verifiable quantum advantage based on polynomials with planted structures
}
\author{Markus Bl\"aser}
\affiliation{Center for Quantum Technologies (QuTe) and Department of Computer Science, Saarland University}
\author{Michael Gullans}
\affiliation{\quera}
 \affiliation{\quics}
\author{Dominik Hangleiter}
\affiliation{\ethz}
\affiliation{\simons}
\author{Yuxuan Liu}
\affiliation{Wuhan University}
\author{Youming Qiao}
\affiliation{Centre for Quantum Software and Information, University of Technology Sydney}
\newcommand{\quics}{Joint Center for Quantum Information and Computer Science, University of Maryland
}
\newcommand{\quera}{QuEra Computing Inc.
}
\newcommand{\simons}{%
  Simons Institute for the Theory of Computing, University of California at Berkeley}
  \newcommand{\ethz}{%
  Institute for Theoretical Physics, ETH Z\"urich}
\begin{abstract}
  A central question in the theory of quantum advantage is whether there are quantum advantage protocols with similar resource requirements as random circuit sampling that are also verifiable just from the classical outputs of the quantum computation. 
  Here, we develop the idea of simulation secrets for verifiable advantage. 
  A verifier can use a simulation secret to evaluate a cross-entropy test faster than it would take a classical adversary to pass the test. 
  We instantiate this idea using IQP circuits described by cubic polynomials with \emph{planted independent spaces}.
  These correspond to the largest independent set in the orbit of a polynomial under the general linear group and yield a low-rank stabilizer decomposition of the corresponding state. 
  We conjecture that large independent spaces are invisible to a computationally bounded adversary, and therefore they cannot exploit them to pass the protocol. 
  A second conjecture regards the fine-grained complexity of producing samples that pass the cross-entropy test for uniformly random polynomials. 
  Under these conjectures, our scheme results in a polynomial gap between the verification time and the time a classical adversary would need to pass the protocol---both are exponential. 
  It has a potential application to generating classically certifiable randomness, since the output distributions have high min-entropy. 
  We estimate that the planted polynomial scheme is implementable using 100 logical qubits at logical error rates around~$10^{-6}$. 
\end{abstract}

\maketitle

\setcounter{tocdepth}{2}

\tableofcontents
\section{Introduction}
\label{sec:introduction}

\subsection{Motivation}

The experimental demonstration of quantum computational advantage has been a major milestone for the field. 
Since the first demonstration in 2019~\cite{arute_quantum_2019}, experiments demonstrating quantum advantage have achieved higher quality and larger circuit volumes \cite{wu_strong_2021,morvan_phase_2024,decross_computational_2025,gao_establishing_2025}. 
A big drawback of those experiments, however, has been that they are not classically verifiable from the experimental outcomes themselves. 
This means that the validation of the advantage demonstrations has relied on separate data \cite{hangleiter_has_2026}. 
This is both practically and  fundamentally significant; 
it is practically significant, because verification has required immense classical computational efforts to simulate computations just below the advantage threshold. It is fundamentally significant, since the lack of verification of the computational outputs themselves gives room for skeptics to doubt the claims, similar to how the first  violations of Bell inequalities \cite{aspect_experimental_1982-1} left room for loopholes, which were only eventually closed much later \cite{giustina_significant-loophole-free_2015,hensen_loophole-free_2015,shalm_strong_2015}. 

The reason why the current demonstrations of quantum advantage are not verifiable is because they rely on the computational task of random circuit sampling \cite{hangleiter_computational_2023}, which refers to the task of sampling from the output distribution of a quantum circuit chosen uniformly from a large family of circuits that is natural on a given architecture. 
The output distributions of those circuits have the so-called anticoncentration property which prevents sample-efficient verification of the correct distribution \cite{hangleiter_sample_2019}. 
The best available way to benchmark random circuit sampling is to use the so-called \emph{cross-entropy benchmark (XEB)} \cite{boixo_characterizing_2018,arute_quantum_2019}, which is the average of the ideal probabilities corresponding to the observed samples. 
This benchmark is sample-efficiently computable and distinguishes ideal samples from random ones.
More broadly, generating samples that achieve a sufficiently high XEB score appears to be a classically hard task \cite{aaronson_complexity-theoretic_2017,gao_limitations_2024,hangleiter_has_2026}. 
However, computing the XEB requires computational effort on a par with simulating an ideal experiment, rendering verification at least as hard as simulating a potentially noisy experiment.  

Theoretically, the best known approaches to classically verifiable quantum advantage are interactive protocols based on post-quantum-secure cryptography \cite{brakerski_cryptographic_2018,brakerski_simpler_2020} and solving the discrete log problem \cite{shor_algorithms_1994} to show quantum advantage \cite{kahanamoku-meyer_jacobi_2024} or even break existing cryptosystems such as  elliptic curve cryptography and RSA \cite{gidney_how_2025,babbush_securing_2026}.
However, the required circuit sizes for these problems are out of reach for next-generation quantum devices \cite{tripier_fault-tolerant_2026} as they require on the order of $10^8$ to $10^9$ Toffoli gates and on the order of 1000 logical qubits \cite{kahanamoku-meyer_classically_2022,babbush_securing_2026,gidney_how_2025}. 

This raises the question if random circuit sampling can be made classically verifiable by the use of cryptography while retaining small circuit sizes that might be achievable on first-generation logical processors with roughly 100 logical qubits and error rates on the order of~$10^{-6}$ \cite{hangleiter_has_2026,tripier_fault-tolerant_2026}. 
One approach, popularized by Aaronson~\cite{aaronson_verifiable_2023}, is to devise circuits which are computationally indistinguishable from random circuits sampled from a  hard-to-simulate family, but have hidden structure in their output distribution that is efficiently detectable. 
In the simplest case, this structure could be a single \emph{secret peak} in the output distribution, corresponding to a single bitstring that is observed with high probability if the prover is honest, but is hard to find for an adversary~\cite{aaronson_verifiable_2024}. 
The verifier, on the other hand, would devise the circuits in such a way that she would be able to \emph{plant} a peak of her choosing when sampling a circuit. 
Using her private knowledge of the peak location she will be able to distinguish an honest prover from an adversarial one.

A compelling candidate construction for this was proposed already by \textcite{shepherd_temporally_2009}, using a family of commuting quantum circuits known as IQP circuits. 
These circuits have simple classical descriptions in terms of a classical linear code, and it turns out that certain properties of the code are reflected in the output distribution---a property exploited by the Shepherd/Bremner protocol. 
Unfortunately, cryptanalysis of the protocol revealed a classical solution~\cite{kahanamoku-meyer_forging_2023}, and even generalizations of the scheme \cite{bremner_instantaneous_2025} can be broken~\cite{gross_secret-extraction_2025}. 
Recent work has built off this idea to devise circuit families with peaks in certain conditional distributions \cite{deshpande_peaked_2025}, and also investigated the broader question of whether peaked circuits that are indistinguishable from random circuits even exist~\cite{aaronson_verifiable_2024,zhang_complexity_2025}.
It remains open whether peaks can be planted to devise a verifiable quantum advantage scheme. 

In this work, we consider a different paradigm for classically verifiable advantage based on \emph{simulation secrets}.
The idea is to plant secret structure in an otherwise random IQP circuit, knowledge of which simplifies classical simulation of the circuit. At the same time, this structure is invisible to a computationally bounded adversary and does not ostensibly alter the properties of its output distribution. 
In particular, the output distribution of the planted circuits has high entropy, as for the random case. 
In this way, the XEB can be used to verify samples since the verifier's private knowledge of the planted structure allows her to compute probabilities significantly faster than what an adversary without knowledge of the secret would be able to do. 

We instantiate this idea with IQP circuits defined by random degree-3 Boolean polynomials \cite{bremner_average-case_2016}. 
The structure we plant is a generalization of planted cliques.  
The gap between simulation time and verification time we obtain in this way is polynomial: 
Under a set of fine-grained complexity conjectures, passing the verification test without knowledge of the secret requires time $\Omega(2^{\kappa n})$ for some constant $\kappa >  1/2$ using polynomial space, while classical verification using the secret only requires time $\tilde O(2^{n(1/2 + \epsilon)})$ and polynomial space for arbitrarily small $\epsilon > 0$. 
This may seem like a modest gap. 
However, with the right choice of parameters it allows for the verification of samples that are not simulatable classically directly from the output of the computation itself. 
What is more, we believe the gap for practical parameter sizes to be significantly larger than the asymptotic conjectures suggest.
We estimate that a classically hard computation using at most 100 000 transversal gates on around 100 logical qubits or less could be verified using this scheme. 
We further conjecture that planted independent spaces remain computationally hidden from quantum algorithms. 

We anticipate that our protocol can be used to generate classically certified randomness. 
This is because the output distribution of our scheme has high min-entropy, and we conjecture that the planted structure remains invisible even to computationally bounded quantum adversaries.
As a first step, we obtain a single-round entropy bound for quantum provers that pass the classical verification test, based on ideas of \textcite{aaronson_certified_2023}. 
Thus, our protocol marks a first step towards a practically useful cryptographic application of next-generation logical quantum computers with around 100 logical qubits \cite{tripier_fault-tolerant_2026}.

\subsection{Quantum circuits from polynomials with planted independent spaces}
\label{subsec:IQP_planted}

Our main conceptual contribution towards this end is to identify and study the \emph{planted independent space} (\PIS) problem for Boolean polynomials. We work with Boolean polynomials of degree at most three, which we call cubic Boolean polynomials throughout this paper.
An independent space of a cubic Boolean polynomial is a linear algebraic analogue of an independent set of a 3-uniform hypergraph, as will be defined and explained later.
Crucially, given a random cubic  polynomial with a planted independent space of dimension $h$, one can simulate the corresponding Boolean IQP circuits in time $2^{n-h}$. 

More specifically, we consider IQP circuits $C(g)$ which are defined by a cubic Boolean polynomial $g$ on $n$-bit strings $x \in \FF_2^n$ (all arithmetic is modulo $2$) 
\begin{align}
  \label{eq:degree-3 poly}
  g(x) = \sum_{i < j < k} G_{ijk} x_i x_j x_k + \sum_{i < j} G_{ij} x_i x_j + \sum_{i\in[n]} G_{i} x_i\, , 
\end{align}
with $G_{ijk }, G_{ij}, G_{i} \in \FF_2$.  Throughout this article, we work with cubic polynomials with constant term $0$ unless otherwise stated. 
Since $x^2=x$ on $\FF_2$, we identify functions $g:\FF_2^n\to\FF_2$ with their unique multilinear representatives in $\FF_2[x_1, \ldots, x_n]$ (see \cref{sub:preliminaries}). 
The corresponding IQP circuit  acts on a computational basis state as 
\begin{align}
  C(g) \ket x = (-1)^{g(x)} \ket x . 
\end{align}
To define the sampling task, we start from a state preparation in the $\ket {+^n}$ state and measure in the $X$ basis, giving rise to outcome amplitudes 
\begin{align}
\label{eq:outcome probs}
 a_g(x) =  \frac 1 {2^{n}}\sum_{y \in \FF_2^n} (-1)^{g(y) + x \cdot y}  \eqqcolon \ngap(g + \langle x, \cdot \rangle ),
\end{align}
where $ \langle x,y \rangle = x \cdot y$ denotes the inner product. 
As we explain below, in \cref{sub:expxeb intro}, we conjecture that there is a $\kappa > 1/2$ such that there is no algorithm that uses polynomial space and has runtime $o(2^{\kappa n})$. 

To motivate independent spaces, let us begin by observing
that a cubic polynomial $g(x) \in\FF_2[x_1, \dots, x_n]$ as in \cref{eq:degree-3 poly} can be associated naturally to a $3$-uniform hypergraph $H(g)$, where each nonzero coefficient $G_{ijk}$ corresponds to a hyperedge $(i,j,k)$. 
Note that this hypergraph only depends on the degree-$3$ monomial coefficients of $g$. 
Recall that in a $3$-uniform hypergraph~$H$, a vertex subset is an independent set, if no hyperedge in~$H$ is contained in it.
Given an independent set of size $h$ in $H(g)$, 
we can simulate the corresponding IQP circuit in time $\tilde O(2^{n-h})$ and polynomial space. 
This is because the complement of an independent set is a vertex cover~$V$ and we can write the outcome amplitudes as 
\begin{align}
  a_g(x) = \frac{1}{2^{|V|}} \sum_{z \in \FF_2^{|V|}} (-1)^{x_V \cdot z}\ngap(g_z + \langle x_{V^c}, \cdot \rangle),
\end{align}
where $g_z$ is a quadratic polynomial in $|V^c|$ variables whose gap is efficiently computable and $x_S$ is the restriction of $x$ to $S \subset [n]$~\cite{maslov_fast_2024,hangleiter_fault-tolerant_2025}.
Here, an independent set of a hypergraph is a vertex subset that contains no 3-ary hyperedges in the hypergraph.
We can use this decomposition for verification: plant a large independent set of size $h = (1-c)n$ in an otherwise uniformly random hypergraph $H(g)$.
The existence and knowledge of such an independent set reduce the simulation time to $2^{cn}$, giving a $\kappa/c$-fold polynomial gap compared to the best unstructured algorithm.

But what is the complexity of finding a planted independent set? 
Independent sets are in one-to-one correspondence to cliques of the complement graph. 
Therefore, finding the planted independent set then amounts to solving the planted clique problem, which has been extensively studied for graphs.
While the maximum clique problem is \np-complete in the worst case setting, planting a clique in the Erd\H{o}s--R\'enyi random graph model $\mathcal{G}(n, 1/2)$ is conjectured to require $n^{\Theta(\log n)}$ time if the planted clique is of size between $(2+\epsilon) \log n$  and $O(\sqrt{n})$ \cite{jerrum_large_1992,kucera_expected_1995,AKS98,HS24}. 
It is of interest as a cryptographic assumption and it can be used for example as a basis of public-key cryptography schemes \cite{applebaum_public-key_2010}, as well as to plant secret backdoors in machine-learning models \cite{goldwasser_planting_2022}.
However, $n^{\Theta(\log n)}$ is not sufficiently hard for our purposes and moreover the problem is efficient for linear-size cliques. 
This is why we require a more sophisticated approach.

To achieve this, we observe that this approach still works even when there is a large independent set somewhere in the \emph{orbit} of $g$ under $\gl(n, \FF_2)$ transformations, that is, the set of hypergraphs induced by the polynomials $g \circ A $ defined as $g \circ A(x) = g(Ax)$ for $ A \in \gl(n, \FF_2)$ and a given cubic polynomial $g$ over $\FF_2$. 
This leads to the following definition. 
\begin{definition}[Independent space] \label{def:ind_sp}
  Let $g:\FF_2^n\to\FF_2$ be a cubic polynomial. We say that $V\leq \FF_2^n$ is an \emph{independent space} of $g$, if there exists  $A \in \gl(n, \FF_2)$ whose first $\dim(V)$ columns  generate $V$ and $g \circ A$ has an independent set on the first $\dim(V)$ coordinates. 
\end{definition}

\begin{figure}
\includegraphics{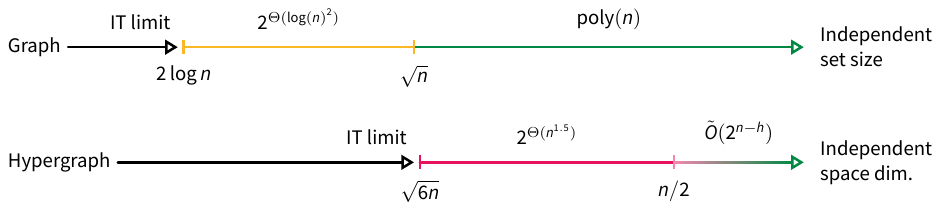}
  \caption{\label{fig:pis vs PIS regimes} Distinguishing regimes for independent sets of random graphs and independent spaces of random hypergraphs. The information-theoretic (IT) limit denotes the typical size of an independent set/subspace. The second transition is the computational boundary where a better algorithm than the trivial one starts to exist. }
\end{figure}

One may think of the problem to decide whether a given cubic polynomial $g$ has a large independent space as an exponentially lifted equivalent of the planted independent set problem wherein the $n$ graph vertices  are replaced by the $2^n$ vectors over $\FF_2^n$. 
In other words, while finding an independent set of size $h$ can be thought of as finding an element $P \in S_n$ of the symmetric group such that $g \circ P$ has an independent set on the first $h$ coordinates, finding an independent space asks for an element of the larger, general linear group $\gl(n, \FF_2)$. 

Analogously to planted independent sets in random graphs, we now consider planted independent spaces in random 3-uniform hypergraphs/random cubic polynomials. 
We find that a uniformly random cubic polynomial has an independent space of dimension $\le \sqrt{6n}$. 
We give evidence that distinguishing a random polynomial from one with a planted (random) $h$-dimensional independent space takes time $2^{\Theta(n^{1.5})}$ for $h \le (1/2 - \epsilon) n$ with $\epsilon > 0 $. 
Just like for planted independent sets, this algorithm checks if there is an independent space of dimension $d = 2 \sqrt{6n}$ by enumerating all these subspaces in time 
$2^{O(n^{1.5})}$
in order to distinguish a given polynomial from a random one. 
For any fixed constant $\epsilon>0$ and all sufficiently large $n$,
when $h\ge(1/2+\epsilon)n$, there exists a nontrivial algorithm with significantly better runtime $\tilde O(2^{n-h})$. 
This algorithm exploits the rank profile of Hessian slices of the degree-3 component, which may be the analogue of the degree distribution of a graph. 
The analogy between (graph) planted independent sets and (hypergraph) planted independent space is shown in \cref{fig:pis vs PIS regimes}. 

Given these tight correspondences between independent sets and independent spaces and the fact that planted independent sets have been extensively studied, we conjecture these limits to be optimal. 
For the purposes of our main result, we make the following conjecture.
\begin{conjecture}[Exponential planted independent space (\Exp\PIS) conjecture]
\label{conj:PIS security intro}
Consider cubic polynomials over $\FF_2$ in $n$ variables.
There is no $o(2^{n})$-time algorithm that distinguishes a random cubic polynomial from a random cubic polynomial with a planted $h$-dimensional independent space, where $h \le (1/2 - \epsilon) n$ for any constant $\epsilon > 0 $.  
\end{conjecture}

By planting an independent space of dimension $\lesssim n/2$ in random cubic polynomials we can obtain a polynomial gap between verification and simulation of samples from the corresponding IQP circuits. 

\subsection{Verifiable advantage from planted independent spaces}
\label{eq:intro verification}

We consider a verifier who chooses a polynomial $g$ uniformly at random with a planted uniformly random independent space of dimension $(1-c)n$ for $ 1/2 < c \le 1 $. 
She now reveals $g$ to an efficient quantum prover, but keeps the planted independent space secret. 
The prover returns samples from running the IQP circuit $C(g)$, and the verifier exploits her secret knowledge of the independent space to verify those samples. 
She then repeats this for polynomially many random choices of $g$ and averages the test outcomes. 

For the verification test, the verifier uses the linear XEB, defined as 
\begin{align}
\label{eq:xeb}
  \chi(q,p) \coloneqq 2^n \sum_{ x \in \FF_2^n} q(x) p(x) - 1, 
\end{align}
for sampled and target probability distributions over $\FF_2^n$,  $q$ and $p$, respectively.
The XEB gives a bona-fide benchmark for \emph{anticoncentrating} probability distributions $p$ whose normalized second moments $2^n\sum_x p(x)^2 = O(1)$ are constant. 
Since the planted independent space keeps the quadratic part of $g$ uniformly random,  this is the case and moreover the variance of the XEB is constant so that it can be sample-efficiently evaluated. 

To evaluate the XEB for a single circuit the verifier therefore computes the ideal probabilities $p(x)$ corresponding to the observed samples $x$ in time $\tilde O(2^{cn})$ using the vertex-cover decomposition algorithm \cite{maslov_fast_2024} and averages the results to obtain an estimate $\hat \chi$ of the XEB. 
She accepts if $\hat \chi > 1$ (the ideal value is $2$) and rejects otherwise. 
An honest prover running a near-ideal version of $C(g)$ will pass this test with high probability. 
On the other hand, we conjecture that there is no algorithm using polynomial space that achieves a high XEB score and has runtime $\tilde O(2^{n/2})$. 
\begin{conjecture}[Exponential XEB conjecture]
\label{conj:exp xeb intro}
 Sampling from a distribution $q$ that achieves $\chi(q,p) \ge 1$ against the ideal output distribution $p$ of a random degree-$3$ IQP circuit with inverse polynomial failure probability requires time $\Omega(2^{\kappa n})$  for some constant $\kappa > 1/2 $ or superpolynomial space. 
\end{conjecture}
We discuss the evidence for \cref{conj:exp xeb intro} below. 
Conjectures~\ref{conj:PIS security intro} and \ref{conj:exp xeb intro} imply our main result, as we show in \cref{sec:verification_of_iqp_circuits}.  
\begin{theorem}[Verifiable advantage from planted independent spaces]
\label{thm:main result}
  Let $g$ be a random cubic polynomial with a planted $(1-c)n$-dimensional independent space for $1/2 < c < \kappa $ and assume Conjectures~\ref{conj:PIS security intro} and \ref{conj:exp xeb intro}.
  Then, 
  \begin{enumerate}[label=(\roman*)]
    \item the XEB of the output distribution of the IQP circuit $C(g)$ can be estimated in time $\tilde O(2^{cn})$ and polynomial space, 
    \item an honest, polynomial-time quantum prover passes the XEB test with high probability, 
    \item there is no classical polynomial-space algorithm with runtime $o(2^{\kappa n})$ that passes the XEB test with high probability. 
  \end{enumerate}

\end{theorem}
In other words, we obtain a polynomial gap of at least $S \sim V^{\kappa/c}$ between verification and simulation times for cubic IQP circuits with $h = (1-c)n$-dimensional planted independent spaces. 
We note that we expect the actual cost differences to be significant in the parameter  regime of $n \lesssim 100$ and $c > 1/2$. 
In this case, verifying samples is feasible using simulation of an at most $50$-qubit IQP circuit, while simulation is beyond the scope of the best classical algorithms on the available hardware.

To the best of our knowledge this constitutes the first example where a significant gap is possible for a random-circuit-like scheme under a plausible cryptographic assumption. 
It is thus the first such example where a quantum computational task that cannot be solved on the best available classical computers can be verified using those same computers.

The key question for the efficacy of our verifiable advantage protocol is the  security of Conjectures~\ref{conj:PIS security intro} and \ref{conj:exp xeb intro}.

\subsection{The planted independent space problem for cubic polynomials}
\label{sub:PIS intro}

The core technical contribution of our work is to define and analyze the complexity of the planted independent space problem. 
We do this in terms of both worst-case and average-case complexity for different parameter regimes.

\begin{problem}[Independent space problem, decision (\decision-\IS)]
Given a cubic polynomial $g\in\FF_2[x_1, \dots, x_n]$ and $h\in[n]$, decide if $g$ admits an independent space of dimension~$h$.
\end{problem}
\begin{problem}[Independent space problem, search (\search-\IS)]
Given a cubic polynomial $g\in\FF_2[x_1, \dots, x_n]$ with the promise that $g$ contains a dimension-$h$ independent space, return a dimension-$h$ independent space. 
\end{problem}

For the definition of the average-case version, we consider the following distributions.
\begin{itemize}
  \item  Let $\mc P_3(n)$ be the uniform distribution over cubic polynomials with constant term $0$ in $\FF_2[x_1, \dots, x_n]$. That is, the coefficients of $x_ix_jx_k$, $x_ix_j$, and $x_k$, are sampled independently and uniformly from~$\FF_2$.
  \item Let $\mc P_3(n,h)$ be the distribution over cubic polynomials obtained by first drawing a uniformly random $g \leftarrow \mc P_3(n)$ and then planting a uniformly random $h$-dimensional independent space~$V$. This can be done by first setting the coefficients of $x_ix_jx_k$, $1\leq i<j<k\leq h$, to be $0$, and then sampling a random $A\in \gl(n, \FF_2)$ and outputting the multilinear representation of $g\circ A$, obtained using $x_i^2=x_i$.
\end{itemize}

\begin{problem}[Planted independent space problem, decision (\decision-$h$-\PIS)]
\label{prob:pis decision}
Given a cubic polynomial $g$ decide whether $g \leftarrow \mc P_3(n)$ or $g \leftarrow \mc P_3(n,h)$. 
\end{problem}

\begin{problem}[Planted independent space problem, search (\search-$h$-\PIS)]
\label{prob:pis search}
Given a cubic polynomial $g \leftarrow \mc P_3(n,h)$, find the $h$-dimensional independent space~$V$. 
\end{problem}

In order to exploit the vertex-cover algorithm for computing outcome probabilities of planted circuits, we need to solve the search version of the \PIS\ problem, and this also solves the decision version.
To give the strongest guarantees, our analysis therefore focuses on the (easier) decision version of the worst- and average-case problems. 

\paragraph{Worst-case hardness and search-to-decision reductions}

We study the complexity of both problems in terms of worst-case hardness as well as in terms of the runtime of concrete algorithms solving either problem. 
As our first result, we determine the worst-case hardness of the decision version of the \IS\ problem. 
\begin{theorem}[Worst-case hardness]\label{thm:np-hard}
\decision-\IS\ is \np-complete.
\end{theorem}
To prove this result, we show tight equivalences between the maximum independent space of a cubic polynomial, the minimum \emph{vertex cover} of its support, its \emph{strength}, and its \emph{symmetric slice rank}. 
Using these equivalences, we can reduce to the problem of deciding whether the slice rank of a given trilinear form is less than a given $r$, which was shown to be \np-hard in \cite{BlaserILPS21}. 

We remark that the independent set problem admits a natural search-to-decision reduction in the worst-case setting. In the average-case setting, the search-to-decision reduction was recently resolved satisfactorily in \cite{HS24}. 

The independent space problem also admits a search-to-decision reduction in the worst-case setting, though such a reduction is not as natural as in the independent set case. That is, by Theorem~\ref{thm:np-hard}, \decision-\IS\ is \np-hard. As it admits natural polynomial-size witnesses (the transformation matrix in $\gl(n, \FF_2)$), it is \np-complete. This allows us to guess the \np-witness (such as a basis of the independent space) bit by bit, and for a partially fixed basis, reduce to the \decision-\IS\ thanks to its \np-completeness. However, this reduction is not natural as the resulting instance of \decision-\IS\ may not relate to the original input instance clearly. 

The average-case search-to-decision reduction for the independent space problem looks rather non-trivial, and we leave this to future work. 

\paragraph{Average-case setting.}

Second, we constrain the region in which we can hope for a polynomial verification-simulation gap by analyzing the complexity of \cref{prob:pis decision,prob:pis search} in the regime where $h \ge n(1/2 + \epsilon)$ for any constant $\epsilon > 0$. In this regime, we give non-trivial algorithms that solve both problems using the rank profile of the Hessian slice matrix evaluated at random vectors. 
\begin{theorem}[The $h \ge n(1/2 + \epsilon)$ regime]
Whenever the planted independent space has dimension~$h\ge(1/2+\epsilon)n$, where $\epsilon>0$ is a constant, \search-$h$-\PIS\ can be solved in time $\tilde O(2^{n-h})$.
\end{theorem}
The key observation of the algorithm achieving this runtime is that the Hessian matrix of a cubic form with a large planted independent space typically has a rank deficiency compared to a uniformly random cubic polynomial at any point in the independent space whenever $h\ge(1/2+\epsilon)n$. 
To distinguish against the uniform case, it is therefore sufficient to find a single element of the independent space. 
To actually find the independent space, a constant number of such points is sufficient since those will pin down the subspace completely, analogously to how large planted cliques of size $s \gtrsim \sqrt n$ can be found quickly once a few of its elements have been found \cite{feige_finding_2010}.

On the other hand, we show that the largest independent space of a random cubic polynomial has dimension $h \le \sqrt{6n}$ with high probability. 
This implies that $\mc P_3(n)$ and $\mc P_3(n,h)$ are statistically distinguishable for $h > \sqrt {6n}$ as outlined above. 
\begin{theorem}
  For $h > \sqrt{6n}$, there is an algorithm solving \decision-$h$-\PIS\ in time $2^{O(n^{1.5})}$. 
\end{theorem}

Based on this, we conjecture that $n/2$ is indeed the computational threshold, yielding \cref{conj:PIS security intro}.

\paragraph{Heuristic algorithms and intuition for Conjecture~\ref{conj:PIS security intro}.}
The \PIS{} problem is yet another addition to the family of planted structure problems. As mentioned, a classical problem in this family is the planted clique problem, which we reviewed above. 

Another related problem is a problem underlying the so-called Unbalanced Oil and Vinegar (UOV) digital signature scheme \cite{patarin_hidden_1996,UOV99}; see \cite{UOV2025} for more recent updates.
This problem concerns the following planted structure. 
Let $\M(n, \FF_q)$ denote the set of $n \times n$ matrices over the field $\FF_q$. 
Given  symmetric matrices $A_1, \dots, A_m\in \M(n, \FF_q)$, plant a random totally-isotropic space $U\leq \FF_q^n$ of dimension $d$; here, $U$ is totally-isotropic if every $A_i$, viewed as a symmetric bilinear form on $\FF_q^n$, evaluates to $0$ on $U$. 

One main approach to attack UOV is to recover this planted totally-isotropic space, and current best algorithms still run in exponential time for $d=m\leq n/3$. This can be done via the Gr\"obner-basis approach by formulating a system of polynomial equations \cite{UOV2025}. To analyze the performance of Gr\"obner-basis attacks, one usually resorts to regularity or semiregularity analyses \cite{BFS15}.
For $h$-\PIS\ in the regime $h \le (1/2 - \epsilon) n$ this prediction is $2^{O(n^{3/2})}$ as suggested by the na\"ive distinguishing algorithm outlined above.

Recently, the planted totally-isotropic space problem was also introduced in another setting in \cite{liu_planted_2026}. 
The work \cite{liu_planted_2026} was motivated by transferring the private simultaneous messages and secret sharing schemes based on planted cliques \cite{abram_cryptography_2023} to a linear algebraic setting. For a better analogy with planted cliques (or equivalently, independent sets), the planted totally-isotropic problem was introduced for alternating matrices. Recall that $A\in\M(n, \FF_q)$ is alternating, if for any $u\in \FF_q^n$, $u^tAu=0$. For $(A_1, \dots, A_m)$, $A_i\in\M(n, \FF_q)$ alternating, the main parameter range considered in \cite{liu_planted_2026} was when $m=\lceil n/\log n\rceil$, and $d\ll n/2$. 

Note that $(A_1, \dots, A_m)$ naturally defines an alternating bilinear map $\phi:U\times U\to W$, where $U=\FF_q^n$ and $W=\FF_q^m$. We can also compare the complexity of solving \PIS\ for polynomials with \PIS\ for alternating bilinear maps $\phi: U \times U \rightarrow W$ satisfying $\phi(V,V) =0$ for a planted subspace $V \le U$ \cite{liu_planted_2026}.
An alternating bilinear map naturally admits an interpretation as a subspace of alternating matrices $\langle M_k \rangle_{k \in [\dim(W)]}$, where $M_k \in V \times V$ are alternating matrices, or, equivalently, a cubic tensor on $U \times U \times W$. 
Whereas an $h$-dimensional independent space of a polynomial corresponds to a basis in which there is an $h\times h \times h$ zero-subcube in the symmetric coefficient tensor, an independent space of a bilinear form corresponds to an $h \times h \times \dim(W)$ \emph{tube} of zeros in the corresponding tensor. 
This interpretation allows for a relatively simple proof of 
\np-hardness of finding the planted subspace by embedding a graph into a tensor using linear independence \cite{bei_independent_2019}.
Correspondingly, in the bilinear-form case, the \PIS\ problem can be solved by asking that $v^T M_k w =0$ for all $v,w \in V$ and $k \in [\dim(W)]$. 
In contrast, in the polynomial/symmetric setting we are in, the condition is the stronger 
\begin{align}
  u^T M_v w = 0 \quad \text{ for all } u,v,w \in V,
\end{align}
where $M_v = \sum_k v_k M_k$ is a certain linear combination of the tensor slices. 
This increase in the degree of the problem makes the na\"ive solution strategies harder. 
Indeed, even in the $h \ge(1/2+\epsilon)n$ setting, the best we can do is an $\tilde O(2^{n-h})$-time algorithm for solving \PIS, while in the same scaling scenario, the alternating bilinear \PIS\ problem admits a polynomial-time solution \cite{liu_planted_2026}.

\subsection{The cubic XEB problem}
\label{sub:expxeb intro}

Let us first consider algorithms for the gap of cubic Boolean polynomials.
This corresponds to an amplitude estimation algorithm, which can be used as a subroutine to generate samples. 
In fact, all known ways to near-exactly sample from IQP circuits involve amplitude estimation as a subroutine \cite{hangleiter_computational_2023}.

Our central piece of evidence for \cref{conj:exp xeb intro} is the fact that there is no known algorithm for computing the gaps of Boolean polynomials that uses runtime $o(2^{n(1-\epsilon)})$ for any $\epsilon >0$ and polynomial space.
The na\"ive and best known polynomial-space algorithm for the Boolean gap just computes the sum of size $2^n$ by going through all $n$-bitstrings in Gray-code order using time $O(2^n n^2)$. 

If we drop the polynomial space requirement, surprisingly, one can improve the na\"ive asymptotic runtime $\Theta(2^n)$ to $\tilde O(2^{0.9965n})$ \cite{lokshtanov_beating_2017,dalzell_how_2020}. 
The key idea of \textcite{lokshtanov_beating_2017} is to find polynomial representations of the partial sums 
\begin{align}
  F_b(y) \coloneqq \sum_{a \in \bin^b}(-1)^{f(y,a)},
\end{align}
for $b = (1-c)n$ computing and storing which requires space scaling strictly better than $2^{n}$. 
Then, in a second step, one can evaluate $F_b$ on all $2^{cn}$ inputs using the compressed representation of $F_b$. 
In upcoming work, we show that this route can be improved significantly, yielding an algorithm with runtime $\tilde O(2^{0.552n + O(\sqrt n)})$ \cite{upcoming}. 
However, there remain barriers to reaching runtime $\tilde O(2^{n/2})$ for arbitrary or random cubic polynomials using this approach: 
we would need to find a size-$2^{n/2}$ representation that lets us {}compute $F_{n/2}$ for every input. 
It also seems as though this approach to computing the gap intrinsically requires a representation of the functions $F_{b}$, and therefore it would not lead to a polynomial-space algorithm. 
From the other side, the best known \ma\ algorithm for computing the gap has runtime $\tilde O(2^{n/2})$ \cite{williams_strong_2016}---and uses a significantly stronger model of computation. 
As we elaborate in the main text, we believe that achieving this runtime in the standard model would be highly surprising. 
This lends further justification to our polynomial-space conjecture.

An alternative route to disprove \cref{conj:exp xeb intro} would be to show that one can pass the XEB using samples from a different distribution that can be simulated more efficiently than the ideal distribution.  
This is the approach that XEB spoofers for noisy circuits take \cite{barak_spoofing_2021,gao_limitations_2024}. 
However, all known spoofers for the XEB exploit noise in the circuit, and there is no known sub-$2^n$ time classical algorithm that achieves an XEB score close to ideal as required in \cref{conj:exp xeb intro}. 
Again, finding such an algorithm would constitute a breakthrough in the theory of XEB.

\subsection{Discussion and outlook}
\label{sub:intro outlook}

Our work opens up several interesting directions of research, both in terms of potential applications and potential improvements. 

\paragraph{Application to classically certified randomness}
An important advantage of our approach to verifiable advantage compared to peaked distributions is that the sample distribution has high entropy. 
This means that they can potentially be used to generate classically certified randomness, an idea that is due to \textcite{aaronson_certified_2023}. 
This yields a useful potential application of our verifiable advantage protocol. 

For the generation of cryptographically certified randomness it is important to take into account that an adversarial prover has access to a quantum computer which they can use to perform efficient quantum computations.
Using a quantum computer, an adversary can pass XEB in time $O(2^{n/2})$ using amplitude amplification while sampling from a distribution with lower min entropy. 
Therefore, the security gap collapses for quantum computers with exponential runtime. 
However, we conjecture that our scheme remains secure against polynomial-time quantum computers, and $O(2^{n/2})$-time and polynomial-space classical computers. 
We take this to be a reasonable adversary model: 
The best a polynomial-time quantum adversary can do is to run the true circuits and return the measurement samples to the classical verifier.
The best a polynomial-space classical adversary can do is to spend time $\tilde O(2^{\kappa n})$ in order to simulate the correct distribution. 
The verifier will then use her knowledge of the planted subspace to verify the samples (in time $\tilde O(2^{cn})$).
This is in contrast with the original proposal using random circuits, where generating certifiable randomness required specific timing assumptions and fast response times of the prover~\cite{liu_certified_2025,liu_certified_2025-1}. 
We sketch this application in \cref{sec:certified randomness}, leaving a full security analysis to future work. 

\paragraph{Towards efficient verification}
One of the central open questions raised by our work is whether our approach can be extended to achieve polynomial-time verification of quantum advantage, that is, an exponential gap between verification and simulation runtimes. 
As we find that polynomials with large independent spaces of dimension $h > n/2$ are distinguishable from random ones in time $\tilde O(2^{n-h})$ and the secret independent space can be recovered in time $\tilde O(2^{n-h})$, 
our scheme hits a barrier at $h = n/2$.
This implies that it can at most yield a polynomial gap between verification and simulation.

Nonetheless, one can think more generally about our approach to plant secret structures in degree-$3$ circuits that simplify simulation of the circuit. 
Several directions remain to be explored.
First, observe that our approach relied on planting a secret low-rank stabilizer decomposition of the output state. 
This makes simulation amenable to stabilizer-rank simulators given knowledge of the secret structure. 
Another natural structure that renders classical simulation more efficient than in the worst case is low entanglement. 
To see how low entanglement manifests in the polynomial coefficients, consider a cubic polynomial $g$ whose coefficients $G_{ijk}=0$ whenever $|k - i| > \log n$  and random otherwise.
The coefficient tensor $G$ then has a \emph{tube} of nonzero entries along the diagonal, and the corresponding circuit has low depth on a 1D line. 
We can again obfuscate the corresponding polynomial as $g = f \circ A$ by a linear map $A \in \gl(n, \FF_2)$. 
We leave it to future work to fully assess the feasibility and security of this approach. 

As a first step, we show in \cref{app:attack_on_block_diagonal_structure} that block-diagonal structure can always be efficiently recovered.
This is a special instance of the low-entanglement structure, where the circuit is unentangled between blocks.
This has applications to a related proposal by \textcite{gheorghiu_obfuscation} for the forrelation problem. The forrelation problem is to decide whether 
\begin{align}
  F(f,g) \coloneqq \langle f, \hat g \rangle  \ge 1/\poly(n), 
\end{align}
given Boolean functions $f,g$ and the Fourier transform 
\begin{align}
  \hat g(s) = \frac 1 {2^n}\sum_x (-1)^{x s}g(x). 
\end{align}
This problem can be solved by a quantum computer since 
\begin{align}
  F(f,g) = \bra {+^n} C(g) H^{\otimes n}C(f) \ket {+^n},
\end{align}
i.e., there is a simple quantum circuit solving the problem. 
\textcite{gheorghiu_obfuscation} consider the forrelation problem for cubic polynomials $f,g:\FF_2^n\to\FF_2$ with a hidden product/block structure. 
They observe that the forrelator $F$ can be efficiently computed classically if the product structure is known, but is classically hard if detecting that structure is hard.
Our recovery algorithm rules out this approach to verifiable quantum advantage. 

\paragraph{Further directions}
Another interesting problem with potential practical applications to verifiable advantage is to consider planted structures under different orbits. Here, we extended the independent set problem corresponding to orbits over $S(n)$, to independent spaces which are orbits over $\gl(n, \FF_2)$. 
Further interesting orbits to consider are Clifford circuits that map cubic polynomial phase states to cubic polynomial phase states, as well as families of efficiently computable and invertible Boolean functions. 

We restricted ourselves to Boolean polynomials over $\FF_2$, but the planted independent space problem is also well defined over arbitrary finite fields $\FF_q$.
The ideas of several results in this article should carry over to fields of higher order, especially for fields of characteristic $>3$ due to the correspondence between symmetric trilinear forms and cubic forms. Still, it would be interesting to examine more carefully how things behave over larger fields.

\subsection{Overview}

The technical part of this paper is structured as follows. 
In \cref{sec:problem statement} we formally introduce polynomials and symmetric tensors, state some basic properties and introduce the planted independent space problem and state our conjecture regarding it. 
In \cref{sec:verification_of_iqp_circuits} we then detail the headline application of the exponential hardness of \PIS, namely, how it can be used for verifiable quantum advantage under an additional conjecture on the exponential difficulty of spoofing the XEB. 
Then, in \cref{sec:worst-case hardness,sec:cryptanalysis} we analyse in detail the worst- and average-case hardness of the \PIS\ problem, respectively.

\section*{AI Statement}
All of the original concepts and proof strategies of this paper are our own. 
We used generative AI of various generations for assistance in proving lemmas, checking technical correctness, and for feedback on earlier versions of the draft. 

\section*{Acknowledgements}

We are grateful for discussions with Michael Walter and Jeroen Zuiddam. 
Part of this work was supported by the Simons Institute for the Theory of Computing, and conducted when MB and YQ were visiting the program ``Complexity and Linear Algebra''.
DH was supported by a Simons postdoctoral fellowship through DOE QSA and NSF QLCI Grant No.\ 2016245, and by the Swiss National Science Foundation through Ambizione Grant No.\ 223764. MJG acknowledges support from NSF QLCI grant OMA-2120757, NSF QLCI grant OMA-2553574, and NSF NQVL:Design:ORAQL grant 2533041.


\section{The polynomial planted independent space problem}
\label{sec:problem statement}

\subsection{Boolean polynomials and tensors}
\label{sub:preliminaries}

We will work over the field $\FF = \FF_2$ and be concerned with cubic Boolean polynomials $f: \FF_2^n \rightarrow \FF_2$ with constant term being $0$ as follows:
\begin{align}
\label{eq:cubic polynomial normal form}
  f(x) = \sum_{1\leq i < j < k \leq n} F_{ijk} x_i x_j x_k + \sum_{1\leq i < j \leq n} F_{ij} x_i x_j + \sum_{1 \le i\le n} F_{i} x_i,
\end{align}
where we understand all addition modulo $2$. 
We decompose the polynomial into unique multilinear degree-$3,2$ and $1$ forms, that is, $f(x)=f^{(3)}(x)+f^{(2)}(x)+f^{(1)}(x)$.
This is natural, as for any $x \in \FF_2$, we have $x = x^k $ for any $k \in \mb N$ and therefore a cubic term of the form $x_i^2 x_j = x_i x_j$ can be reduced to a quadratic term. 
Throughout, we identify every polynomial $f:\FF_2^n\to\FF_2$ with its unique multilinear representative in $\FF_2[x_1,\dots,x_n]$, i.e., we work in $\FF_2[x_1,\dots,x_n]/(x_i^2-x_i)$. 
When referring to degrees, coefficients and (formal) derivatives we refer to this representative, and we reduce compositions such as $f\circ A$ back to this representative.
A homogeneous degree-$3$ polynomial $f$ is also called a \emph{cubic form}. 

In the following, we record some structural results about cubic polynomials over $\FF_2$. While these should be known in the literature, we prove them for completeness.

\paragraph{Hessian matrices of cubic polynomials.}
Let $f\in \mathbb{F}_2[x_1, \dots, x_n]$
be a cubic polynomial as in \cref{eq:cubic polynomial normal form}.
\begin{equation}\label{eq:partials}
\partial_i\partial_j f=f(x+e_i+e_j)+f(x+e_i)+f(x+e_j)+f(x).
\end{equation}
\begin{fact}\label{fact:partials}
\begin{enumerate}
    \item For $i, j\in[n]$, $\partial_i\partial_j f=\partial_j\partial_i f$.
    \item When $i=j$, $\partial_i\partial_j f=0$.
    \item When $i<j$, $\partial_i\partial_j f=\sum_{1\leq k<i}F_{k,i,j}x_k+\sum_{i<k<j}F_{i,k,j}x_k+\sum_{j<k\leq n}F_{i,j,k}x_k+F_{i,j}$. 
\end{enumerate}
\end{fact}
\begin{proof}
    Items (1) and (2) are straightforward. For (3), it is enough to check those monomials in $f$. For degree-$3$ monomials $x_ax_bx_c$ where $1\leq a<b<c\leq n$, we distinguish between 
    \begin{itemize}
        \item $\{i, j\}\subseteq \{a, b, c\}$. For example, for $M=x_ix_jx_c$, then $M(x+e_i+e_j)+M(x+e_i)+M(x+e_j)+M(x)=(x_i+1)(x_j+1)x_c+(x_i+1)x_jx_c+x_i(x_j+1)x_c+x_ix_jx_c=x_c$.  
        \item $i$ or $j$ is in $\{a, b, c\}$. For example, for $M=x_ix_bx_c$, then $M(x+e_i+e_j)+M(x+e_i)+M(x+e_j)+M(x)=0$ as $e_j$ is irrelevant in $M$. 
        \item Neither $i$ nor $j$ is in $\{a, b, c\}$. In this case the evaluation is zero as above.
    \end{itemize}
    Then among quadratic terms, only $x_ix_j$ would contribute $1$, and linear terms all contribute to $0$. The result then follows.
\end{proof}

\begin{definition}\label{def:hessian}
    Let $f\in\FF_2[x_1, \dots, x_n]$ be a cubic polynomial. 
    
    The Hessian of $f$, $\Hess^f$, is an $n\times n$ matrix whose $(i, j)$th entry $\Hess^f(i, j)=\partial_j\partial_if(x)$. 
    
    Let $\Hess^f_v$ be the alternating matrix whose $(i, j)$th entry $\Hess^f_v(i, j)=\partial_j\partial_if(v)$. More specifically, 
    \begin{equation}\label{eq:Mij}
    \Hess^f_v(i, j) = \partial_j \partial_i f(v) = f(v + e_i + e_j) + f(v + e_i) + f(v + e_j) + f(v).
    \end{equation}

    The matrix associated with the quadratic part of $f$ is the $n\times n$ alternating matrix  $\Quad^f$ defined entrywise as $\Quad^f(i,j)= \partial_i f^{(2)}(e_j) $, which evaluates to $\Quad^f(i, j)=F_{i,j}$ for $1\leq i<j\leq n$, $\Quad^f(i, j)=F_{j, i}$ for $1\leq j<i\leq n$, and $\Quad^f(i,i)=0$ for $i\in[n]$.
\end{definition}

\begin{fact}
Let $f$, $\Hess^f$, and $\Quad^f$ be from Definition~\ref{def:hessian}. Then $\Hess^f_{u+v}+\Quad^f=(\Hess^f_u+\Quad^f)+(\Hess^f_v+\Quad^f)$. That is, the matrices $\{\Hess^f_u+\Quad^f \mid u\in\FF_2^n\}$ form a linear space.
\end{fact}
\begin{proof}
By Fact~\ref{fact:partials}, the entries in $\Hess^f$ are linear polynomials whose constant terms are $F_{i, j}$, the entries of the matrix $\Quad^f$. As a result, $\Hess^f_{u+v}-\Hess^f_u-\Hess_v^f=\Quad^f$.
\end{proof}

\paragraph{From cubic polynomials to trilinear forms} Let $T:\FF_2^n\times\FF_2^n\times\FF_2^n\to\FF_2$ be a trilinear form. We say that $T$ is alternating, if whenever two arguments are the same, $T$ evaluates to $0$; that is, for any $u, v\in\FF_2^n$, $T(u, u, v)=T(u, v, u)=T(v, u, u)=0$.
\begin{definition}\label{def:Tf}
Let $f: \mathbb{F}_2^n \to \mathbb{F}_2$ be a cubic polynomial as in \cref{eq:cubic polynomial normal form}. 
The alternating trilinear form $\TF^f: \mathbb{F}_2^n \times \mathbb{F}_2^n \times \mathbb{F}_2^n \to \mathbb{F}_2$  associated with the cubic part of $f$ is obtained by specifying its evaluations on the standard basis vectors $\{e_i\}_{i=1}^n$ and extending by linearity as follows: 
\begin{enumerate}
    \item For distinct indices $i< j< k$ and any permutation $\pi \in S_3$, we set $\TF^f(e_{\pi(i)}, e_{\pi(j)}, e_{\pi(k)}) = F_{ijk}$. 
    \item If any indices coincide, $\TF^f$ evaluates to $0$. 
\end{enumerate}
\end{definition}

With this formalization, we can establish the basic relationship between the cubic polynomial $f$ and its associated trilinear form $\TF^f$ via the following identity.

\begin{lemma}\label{lem:tensor}
Let $f: \mathbb{F}_2^n \to \mathbb{F}_2$ be a cubic polynomial and $\TF^f$ be its associated alternating trilinear form as in Definition~\ref{def:Tf}. Then for any vectors $u, v, w \in \mathbb{F}_2^n$, 
\begin{equation}\label{eq:Tf}
\TF^f(u, v, w) = f(u+v+w) + f(u+v) + f(v+w) + f(u+w) + f(u) + f(v) + f(w).
\end{equation}
\end{lemma}
\begin{proof}
We first consider the cubic monomials. In Equation~\ref{eq:Tf}, 
\begin{align*}
RHS = \sum_{1\leq a<b<c\leq n} F_{abc} \Big( 
&(u_a+v_a+w_a)(u_b+v_b+w_b)(u_c+v_c+w_c)+ \\
&(u_a+v_a)(u_b+v_b)(u_c+v_c) + \\
&(v_a+w_a)(v_b+w_b)(v_c+w_c) + \\
&(u_a+w_a)(u_b+w_b)(u_c+w_c)+ \\
&u_a u_b u_c + v_a v_b v_c + w_a w_b w_c \Big).
\end{align*}

For a fixed unordered set $\{a, b, c\}$, we expand the expression and count the terms based on the number of components they contain from the vector $u$:
\begin{itemize}
  \item 
The term $u_a u_b u_c$ appears in the expansion of $(u_a+v_a+w_a)(u_b+v_b+w_b)(u_c+v_c+w_c)$ once; appears in the expansion of $(u_a+v_a)(u_b+v_b)(u_c+v_c)$ once; appears in the expansion of $(u_a+w_a)(u_b+w_b)(u_c+w_c)$ once; appears in the expansion of $u_a u_b u_c$ once. It occurs exactly $4 \equiv 0 \pmod 2$ times and cancels out completely.

\item 
The term $u_a u_b v_c$ (which contains two components of $u$) appears in the expansion of $(u_a+v_a+w_a)(u_b+v_b+w_b)(u_c+v_c+w_c)$ once; appears in the expansion of $(u_a+v_a)(u_b+v_b)(u_c+v_c)$ once. It occurs exactly $2 \equiv 0 \pmod 2$ times and cancels out completely.

\item 
The term $u_a v_b w_c$ (containing exactly one $u$, one $v$, and one $w$ component) only appears in $(u_a+v_a+w_a)(u_b+v_b+w_b)(u_c+v_c+w_c)$.
\end{itemize}
Therefore, after cancellation, we have
$$
RHS=\sum_{a<b<c} F_{abc} ( u_a v_b w_c + u_a w_b v_c + v_a u_b w_c + v_a w_b u_c + w_a u_b v_c + w_a v_b u_c ).
$$
On the other hand, 
$$
LHS=\TF^f(u, v, w) = \sum_{i=1}^n \sum_{j=1}^n \sum_{k=1}^n u_i v_j w_k \TF^f(e_i, e_j, e_k).
$$
Since $\TF^f(e_i, e_j, e_k)$ is non-zero only if the indices are mutually distinct, the triple summation essentially iterates over all triplets of distinct indices, summing the values over their $3! = 6$ permutations. By rewriting the left-hand side using the subset condition $a<b<c$, we obtain:
\[
\TF^f(u, v, w) = \sum_{a<b<c} F_{abc} \left( \sum_{\pi \in S_3} u_{\pi(a)} v_{\pi(b)} w_{\pi(c)} \right).
\]

We then deal with quadratic and linear terms on the right-hand side of Equation~\ref{eq:Tf}, and it can be verified that these contribute $0$, which concludes the proof.
\end{proof}

\begin{theorem}\label{cor:equiv}
    Let $f\in\FF_2[x_1, \dots, x_n]$ be a cubic polynomial. For $u\in \FF_2^n$, let $\TF^f_u:\FF_2^n\times\FF_2^n\to\FF_2$ be defined as $\TF^f_u(v, w)=\TF^f(u, v, w)$. Then $\TF^f_u=\Hess^f_u+\Quad^f$. 
\end{theorem}

\begin{proof}
    By substituting $v = e_i$ and $w = e_j$ into Equation~\ref{eq:Tf}, we obtain
    \[
    T^f(u, e_i, e_j) = f(u+e_i+e_j) + f(u+e_i) + f(e_i+e_j) + f(u+e_j) + f(u) + f(e_i) + f(e_j).
    \]

    We first note that in either the above equation or the definitions of $\Hess^f_u$ and $\Quad^f$, the linear term vanishes by linearity over $\FF_2$. This leaves us to consider only cubic and quadratic terms, which in turn indicate that vectors of weights $\leq 1$ evaluate to zero. Therefore, the $(i, j)$th entry of $\TF^f_u$ is 
     \[
    \TF^f(u, e_i, e_j) = f(u+e_i+e_j) + f(u+e_i) + f(e_i+e_j) + f(u+e_j) + f(u).
    \]
    Now note that $f(u+e_i+e_j) + f(u+e_i)  + f(u+e_j) + f(u)$ is equal to the $(i, j)$th entry of $\Hess^f_u$, and $f(e_i+e_j)$ is equal to the $(i, j)$th entry of $\Quad^f$. This concludes the proof. 
\end{proof}

\begin{proposition}
Let $f:\FF_2^n\to\FF_2$ be a cubic polynomial. For $A\in\gl(n, \FF_2)$, let $g=f\circ A$.
Then we have the following. 
    \begin{enumerate}
        \item $(f\circ A)^{(3)}=(f^{(3)}\circ A)^{(3)}=g^{(3)}$.
        \item For all $u,v,w\in \FF_2^n$,
        $\TF^g(u,v,w)=\TF^f(Au,Av,Aw)$.
        \item As matrices of bilinear forms, $\TF^g_u=A^t \TF^f_{Au}A$. Consequently, \[ \Hess^g_u+\Quad^g=A^t(\Hess^f_{Au}+\Quad^f)A. \]
    \end{enumerate}
\end{proposition}
\begin{proof}
For (1), this is because the degree-$3$ part of $f\circ A$ only depends on $f^{(3)}$. For (2), this is due to Lemma~\ref{lem:tensor}, which expresses $\TF^f$ as evaluations by $f$. (3) is a consequence of (2) and Theorem~\ref{cor:equiv}.
\end{proof}

\begin{remark}[Tensor Slices Perspective]
The reason for introducing $\TF^f$ is that a trilinear form allows for specializing one argument to obtain bilinear forms, where ranks and normal forms are usually well-studied. This strategy is in line with associating quadratic forms with bilinear forms in the case of field characteristic not $2$. It can be implemented relatively easily for trilinear forms over fields of characteristic not $2$ or $3$: for a cubic form $f$, the associated symmetric trilinear form is defined as
$$
\TF^f(u, v, w) = \frac{1}{6}\big(f(u+v+w) - f(u+v) -f(v+w) - f(u+w) + f(u) + f(v) + f(w)\big).
$$

What we observed in the above is that over $\FF_2$, a cubic polynomial can be associated with an alternating trilinear form defined in Definition~\ref{def:Tf} or alternatively by Equation~\ref{eq:Tf} which essentially captures the cubic coefficients of $f$ and is unaffected by its quadratic and linear terms.
\end{remark}

\subsection{Planted independent space}

In \cref{subsec:IQP_planted}, independent spaces were defined
using independent sets after a linear change of variables,
where independence concerns only the cubic terms. The following
proposition gives an equivalent characterization in terms of
the associated alternating trilinear form.

\begin{proposition}
Let $g:\FF_2^n\to\FF_2$ be a cubic polynomial and let
$V\leq\FF_2^n$. Then $V$ is an independent space of $g$,
as defined in \cref{def:ind_sp}, if and only if $\TF^g|_{V^3}\equiv 0$, i.e., $\TF^g(u,v,w) = 0$ for all $ u,v,w\in V$, or equivalently, $\deg (g|_V) \le 2$.
\end{proposition}

\begin{proof}
Let $d=\dim V$. By \cref{def:ind_sp}, choose $A\in\gl(n, \FF_2)$ whose first $d$ columns form a basis of $V$. In the multilinear representation of $g\circ A$, the coefficient of $x_i x_j x_k$, for $1\le i<j<k\le d$, is
\[
\TF^{g\circ A}(e_i,e_j,e_k)=\TF^g(Ae_i,Ae_j,Ae_k).
\]
The first $d$ coordinates form an independent set of $g\circ A$ if and only if all these coefficients vanish. Since $Ae_1,\ldots,Ae_d$ form a basis of $V$, and $\TF^g$ is alternating and trilinear, this is equivalent to $\TF^g|_{V^3}\equiv 0$.
\end{proof}


We further say that $V$ is a \emph{maximum independent space} if 
  \begin{align}
    V \in \arg\max_{ W \le \FF_2^n: \TF^g|_{W^3}=0 } \dim(W)
  \end{align}
  We use $\alpha(g)$ to denote the maximum dimension over independent spaces of $g$.

Independent spaces have a natural connection to independent sets in the hypergraph associated with a polynomial. 
\begin{definition}[Polynomial hypergraph]
\label{def:poly hypergraph}
Given a cubic polynomial $f \in \FF_2[x_1, \ldots, x_n]$, the hypergraph $H(f) = (V,E)$ is defined by the trilinear coefficients $F_{ijk}$ as follows: 
Let $V = [n]$ be the vertex set, and let the set of hyperedges be $E = \{\{i,j,k\}: F_{ijk} =1, i < j < k \}$. 
\end{definition}

\begin{remark}[Independent set]
Consider a cubic polynomial $f\in \FF_2[x_1, \dots, x_n]$. 
For $I\subseteq[n]$, $I$ is a (strong) independent set of $H(f)$ if and only if $\spn\{e_i:i\in I\}$ is an independent space of $f$.
\end{remark}
Moreover, notice that finding a maximum independent set is equivalent to finding a maximum clique in the complement graph and therefore computationally equivalent. 

\paragraph{Random cubic polynomials.} A random cubic polynomial in $\FF_2[x_1, \dots, x_n]$ is sampled by sampling the coefficient of every term $x_i x_j x_k$ ($1\leq i<j<k\leq n$), $x_ix_j$ ($1\leq i<j\leq n$), $x_i$ ($i\in[n]$) independently uniformly from $\FF_2$. 
\begin{lemma}
    Let $f\in \FF_2[x_1, \dots, x_n]$ be a random cubic polynomial. Then $\alpha(f)\leq \lfloor\sqrt{6n}\rfloor$ with high probability. 
\end{lemma}
\begin{proof}
    This is a standard probabilistic method. For a dimension-$d$ space $U\leq \FF_2^n$, the probability of $U$ being an independent space for a random $f$ is $1/2^{\binom{d}{3}}$. The number of $d$-dimensional spaces of $\FF_2^n$ is upper bounded by $2^{d(n-d)+2}$. By a union bound, over random choices of $f$, 
    \begin{equation*}
        \Pr[\exists U\leq \FF_2^n, \dim(U)=d, T^f|_{U^3}=0] 
        \leq  \frac{2^{d(n-d)+2}}{2^{\binom{d}{3}}} =\frac{1}{2^{\frac{d(d-1)(d-2)}{6}-d(n-d)-2}}.
    \end{equation*}
    Note that when $d\geq \lceil \sqrt{6n}\rceil$, $d(d-1)(d-2)/6-d(n-d)-2\geq d^2/2\geq 3n$, so 
    $$\Pr[\exists U\leq \FF_2^n, \dim(U) =d, \allowbreak T^f|_{U^3}=0]\ll 1/2^{3n},$$ showing the existence of $f$ without dimension-$d$ independent spaces. 
\end{proof}


\subsection{Statement of independent space problems}
\label{sub:statement}

Recall the distributions $\mc P_3(n)$ and $\mc P_3(n,h)$, as well as the decision and search problems introduced in \cref{sub:PIS intro}. 

We now describe the planted sampling procedure for $\mc P_3(n, h)$ in detail. Then in this section, we formally specify the planted independent space problems and discuss the relations between these problems.

Fix an integer $1\le h\le n$ and let $U=\operatorname{span}\{e_1,\ldots,e_h\}$. 
We first sample a cubic Boolean polynomial $f\leftarrow\mc P_3(n)$. We then plant an independent
space in its cubic component by setting the coefficient
of $x_i x_j x_k$ to zero whenever $1\le i<j<k\le h$,
leaving all other coefficients unchanged.
Write the resulting polynomial as $\widetilde f=c_0+r_0$, where $c_0$ is its homogeneous cubic component and $r_0$ is its independently sampled quadratic and linear
part. Thus, $\TF^{\widetilde f}|_{U^3}=0$.

To hide this coordinate subspace, we independently sample $A\leftarrow\mathrm{GL}(n,\FF_2)$ uniformly and set
\[
    g(x)=\widetilde f(Ax),
    \qquad
    V=A^{-1}U,
\]
and then reduce $g$ to its unique multilinear representation using $x_i^2=x_i$. Then $V$ is a uniformly random $h$-dimensional subspace satisfying $\TF^g|_{V^3}=0$, or equivalently, $\deg(g|_V)\le2$. We denote the resulting distribution of $g$ by $\mc P_3(n,h)$.

The linear substitution may produce quadratic and linear terms from $c_0$. These terms are added to $r_0\circ A$, and all resulting coefficients are retained in the public polynomial $g$. Since composition with $A$ is a bijection on the space of Boolean polynomials of degree
at most two with zero constant term, $r_0\circ A$ is uniformly distributed even after fixing $c_0$ and $A$. Adding the lower-degree contribution from $c_0\circ A$ preserves this uniformity. Consequently, the final lower-degree part $g^{(2)}+g^{(1)}$ is uniformly random and independent of the pair $(g^{(3)},V)$.

The public instance consists of $n$, $h$, and the full coefficient description of $g$, while $V$ and $A$ are kept secret. 
The following computational problems associated with finding maximum planted independent spaces are natural generalizations of finding maximum or planted independent sets or cliques.

\begin{problem}[Independent space problem (\decision-\IS)]
Given a cubic Boolean polynomial $g$ and an integer $h$, decide whether $g$ has an independent space of dimension at least~$h$. 
\end{problem}

Next, we introduce the distributional version of this problem, where one is asked to distinguish between a random polynomial and one with a planted independent space.

\begin{problem}[Planted independent space problem (\decision-$h$-\PIS)]
\label{prob:pis decision-full}
Given a cubic Boolean polynomial $g$ decide whether $g \leftarrow \mc P_3(n)$ or $g \leftarrow \mc P_3(n,h)$, promised one is the case. 
\end{problem}

Finally, we define the associated search problems, which ask to \emph{find} an independent space of a given dimension $h$.

\begin{problem}[Independent subspace problem (\search-\IS)]
Given a cubic Boolean polynomial $g$ with an independent space~$V$ of dimension $h$, find $V$.
\end{problem}

\begin{problem}[Planted independent space problem (\search-$h$-\PIS)]
\label{prob:pis search-full}
Given a cubic Boolean polynomial $g \leftarrow \mc P_3(n,h)$, find the $h$-dimensional independent space $V$. 
\end{problem}

Our central problem is the \search-$h$-\PIS\ problem, and there are some obvious relations between the hardness of these problems 
\begin{align}
  \text{\decision-$h$-\PIS} &\le_p \text{\search-$h$-\PIS} \le_p \text{\search-\IS}\\
 \text{\decision-$h$-\PIS} &\le_p \text{\decision-\IS} \le_p \text{\search-\IS}, 
\end{align}
where $\le_p$ represents polynomial-time reductions. 
However, it remains an interesting open problem whether there are search-to-decision reductions for the independent space problems, analogous to those for planted clique~\cite{HS24}.

\paragraph{Hardness conjecture}

The basis of our cryptographic applications of the \PIS\ problem will be the following conjecture about the computational complexity of finding a large planted independent space. 
\begin{conjecture}[Exponential planted independent space (\Exp\PIS) conjecture] \label{conj:PIS}
For any $\epsilon > 0 $ and $h \le (1/2 - \epsilon) n$, every  classical algorithm $\mc A$ that outputs $1$ for planted instances and $0$ for uniform instances and has runtime $o(2^{n})$ satisfies 
   \begin{align}
     \left|\Pr_{g\leftarrow\mc P_3(n,h)}
       [\mc A(g,h)=1]
    -
    \Pr_{g\leftarrow\mc P_3(n)}
       [\mc A(g,h)=1] \right| = \negl(n),
   \end{align}
    where the probabilities include the internal randomness of $\mc A$. 
\end{conjecture} 

We provide evidence for \cref{conj:PIS} in \cref{sec:worst-case hardness,sec:cryptanalysis} below by analysing its worst- and average-case complexity in different parameter regimes. 
We note that in fact we believe the \PIS\ problem to be exponentially hard even for quantum computers.
This is relevant to our envisioned application of planted independent spaces for classically certifiable randomness. We leave a more detailed analysis to future work, however. 

\section{Verification of IQP circuit sampling from planted independent spaces }
\label{sec:verification_of_iqp_circuits}

In this section, we introduce the family of IQP circuits, and show how the planted independent space conjecture allows us to obtain a verifiable quantum advantage scheme. 
We first introduce IQP circuits and their statistical properties in \cref{sub:iqp_circuits_and_their_properties} and  then show how the XEB of outcome samples of IQP  circuits is used to verify quantum advantage in \cref{sub:verification_via_xeb}. 
In \cref{sub:simulating_iqp_circuits_with_and_without_pis} we show how planted independent spaces are \emph{simulation secrets} that simplify verification. Finally, in \cref{sec:certified randomness}, we sketch an application of planted IQP circuit sampling to classically certifiable randomness. 

\subsection{IQP circuits and their statistical properties}
\label{sub:iqp_circuits_and_their_properties}

Let $f(x) \in \FF_2[x_1, \ldots, x_n]$ be a cubic Boolean polynomial as in \cref{eq:cubic polynomial normal form}. 
The associated degree-$3$ \emph{IQP circuit} $C(f)$ is defined by its action on basis states as
\begin{align}
\label{eq:def iqp circuit}
  C(f) \ket x = (-1)^{f(x)} \ket x. 
\end{align}
Degree-$3$ IQP circuits are therefore diagonal quantum circuits composed of $Z$, $CZ$, and $CCZ$ gates acting on the qubits specified in the coefficients $F_i, F_{ij},$ and $F_{ijk}$, respectively. 
Applying an IQP circuit $C_f$ to an initial $\ket {+^n}$ state yields the \emph{polynomial phase state}
\begin{align}
  \label{eq:phase state}
  \ket f \coloneqq C(f) \ket {+^n} = \frac 1 {\sqrt{2^n}}\sum_x (-1)^{f(x)}\ket x. 
\end{align}
The amplitudes of $\ket f$ in the Hadamard basis 
\begin{align}
\label{eq:def iqp amplitude}
  a_f(x) \coloneqq \bra x H^{\otimes n} \ket {f},
\end{align}
 yield the corresponding probabilities $p_f(x) = a_f(x)^2$ of observing $x$ when measuring in this basis. 

IQP circuits have been extensively studied in the context of quantum advantage demonstrations. 
This is because they are a subuniversal family of quantum circuits whose output distributions are yet hard to simulate classically assuming plausible complexity-theoretic conjectures~\cite{shepherd_temporally_2009,bremner_classical_2010,bremner_average-case_2016,bremner_achieving_2017,hangleiter_fault-tolerant_2025}. 
The key observations leading to this result are that the output amplitudes $a_f(x) = \ngap(f_x)$ are given by the gap of the Boolean function $f_x(y) = f(y) + \langle x, y\rangle $, defined as 
\begin{align}
  \ngap(f) \coloneqq \frac 1 {2^n} \left ( \left| \{ x: f(x) = 0  \}\right| - \left|\{x: f(x) = 1\}\right| \right). 
\end{align}
Computing $a_f(x)$ up to a constant relative error is therefore a \gapp-hard problem, where \gapp\ is the class of counting problems that computes the difference of two \sharpP\ functions \cite{goldberg_complexity_2017}. 

\begin{theorem}[\cite{bremner_average-case_2016}]
  Let $f$ be a uniformly random degree-$3$ polynomial. Then there exists  $\epsilon = O(1) $ such that producing samples from any distribution $q$ that is $\epsilon$-close in total-variation distance to $p_f$ is classically intractable assuming that   
  \begin{enumerate}[label=\roman*.]
    \item \label{conj:1} the polynomial hierarchy does not collapse to its third level, and
    \item \label{conj:2} the outcome probabilities $p_f(x)$ are \gapp-hard to approximate up to constant relative error for a constant fraction of the $f$-instances. 
  \end{enumerate}
\end{theorem}
While condition \ref{conj:1} is widely believed in complexity theory~\cite{arora_computational_2009}, condition \ref{conj:2}---while commonly regarded as highly plausible---is more specific to the literature on quantum advantage. 
The arguments for its truth are reviewed in detail in Ref.~\cite{hangleiter_computational_2023}. In short, they boil down to: 
(a) there is no exploitable structure in a random degree-$3$ polynomial that would simplify the computation compared to the worst case, and 
(b) a large fraction of the outcome probabilities are comparable in size, a property known as anticoncentration, 
\begin{align}
\label{eq:anticoncentration}
  \Eb_{f}\left[\sum_x p_f(x)^2 \right] \le \frac {\alpha}{2^n},
\end{align}
so that roughly the same absolute error must be achieved on a large fraction of the  probabilities.

Since anticoncentration also plays an important role in verification, an important fact about IQP circuits is going to be the first few moments of their output distributions. 
\begin{itemize}
  \item 
Let $\mc P_d(n)$ be the uniform distribution over polynomials with degree at most $d$ in $\FF_2[x_1, \ldots, x_n]$ with constant term $0$, i.e., the coefficient of each multilinear monomial of degree between $1$ and $d$ is sampled independently and uniformly from $\FF_2$.
\item Let $\mc M(n)$ be the distribution over cubic Boolean polynomials in $\FF_2[x_1, \ldots, x_n]$ obtained by defining $f = g + f^{(3)}$, where $g \leftarrow \mc P_2(n)$ and $f^{(3)}$ is drawn from an arbitrary distribution $\mc D$  over cubic polynomials. 
\end{itemize}

\begin{lemma}[Moment bounds for IQP circuits]
\label{lem:moments iqp}
For $k = 1,2,3$, the normalized moments of $\mc M(n)$ are bounded as 
  \begin{align}
    m_k \coloneqq 2^{nk} \Eb_{f \leftarrow \mc M(n)} p_f^{k}(0)= (2k-1)!! + o(1), 
  \end{align}
  where $(2k-1)!! = 1 \cdot 3 \cdot 5  
  \cdots (2k-1)$.
\end{lemma}
\begin{proof}
Consider the special case that $\mc D $ is the point measure on a fixed cubic form $f^{(3)}$, i.e., a cubic polynomial with no linear or quadratic terms.  
We follow Eq.~(42) of \textcite{dalzell_how_2020} and find that 
\begin{align}
  2^{2nk} \Eb_{f \leftarrow \mc M(n)} \left[a_f(x) ^{2k} \right] & =  \sum_{x^1, \ldots, x^{2k}} \Eb_{g \leftarrow \mc P_2(n)} (-1)^{(g + f^{(3)})(x^1) + \ldots + (g + f^{(3)})(x^{2k})}\\ 
  & = \sum_{x^1, \ldots, x^{2k}} (-1)^{f^{(3)}(x^1) + \ldots + f^{(3)}(x^{2k})} \Eb_{g \leftarrow \mc P_2(n)} (-1)^{g(x^1)+ \ldots + g(x^{2k})} 
  \\
  & = \sum_{x^1, \ldots, x^{2k}} (-1)^{f^{(3)}(x^1) + \ldots + f^{(3)}(x^{2k})} \prod_{i,j=1}^n \id (x_i^1 x_j^1 + \ldots + x_i^{2k} x_j^{2k} = 0)\\ 
   & =  \sum_{x^1, \ldots, x^{2k}}\prod_{i,j=1}^n \id (x_i^1 x_j^1 + \ldots + x_i^{2k} x_j^{2k} = 0) \label{eq:third order contrib}\\ 
  & = 2^{2nk} \Eb_{g \leftarrow \mc P_2(n)}\left[ a_g(x)^{2k}\right], 
\end{align}
i.e., the first three moments of $\mc M(n)$ are given by the moments of $\mc P_2(n)$. 

To see why \cref{eq:third order contrib} holds, write the matrix $X = (x_i^l)_{l,i}$, and observe that the product of indicators enforces that $X(i)X(j) =0$ for all $i,j \in [n]$, that is, all columns of $X$ are orthogonal and have even Hamming weight. 
We can then write 
\begin{align}
  \sum_{l=1}^{2k}f^{(3)}(x^l) &= \sum_{l=1}^{2k}\sum_{i < j < k}a_{ijk} x^{l}_i x^{l}_jx^{l}_k \eqqcolon \sum_{i < j < k} a_{ijk} \langle X(i),X(j), X(k) \rangle,
\end{align}
where we have defined $\langle u,v,w \rangle = \sum_{l=1}^{m} u_l v_l w_l \mod 2$ for $u,v,w \in \FF_2^{m}$ and notice that $\langle u,v,w \rangle = |\supp(u) \cap \supp(v) \cap \supp(w)| \mod 2$. 
We now observe that any self-orthogonal subspace $C \le \FF_2^{m}$ such that $C \subset C^\perp $ with $m \le 6$ satisfies 
\begin{align}
  \langle u,v,w \rangle \equiv 0. 
\end{align}
To prove this, we partition the set $[m]$ into the seven disjoint regions $t = \supp(u)\cap \supp(v)\cap \supp(w)$, $p_{uv} = \supp(u)\cap \supp(v) \setminus t$, and $s_u = \supp(u) \setminus t \setminus p_{uv} \setminus p_{vw}\setminus p_{uw}$. Self-orthogonality implies 
\begin{align}
  \langle u, v \rangle = |t|+ |p_{uv}| = 0  \quad \Rightarrow \quad |p_{uv}| = |t| \\ 
\langle u, u \rangle = |t| + |p_{uv}| + |p_{uw}| + |s_u|= 0  \quad \Rightarrow \quad| s_{u}| = 3t  ,
\end{align}
and therefore all regions have parity $|t|$. But for $|t| = \langle u,v,w \rangle $ to be odd, $m \ge 7$ is necessary.

We can now apply the bounds from Theorem 8 by \textcite{dalzell_how_2020}, using their subsequent observation for degree-2 polynomials, to obtain the claimed result. 

The generalization to arbitrary distributions $\mc D$ over cubic polynomials is immediate: for arbitrary cubic polynomials, just decompose them into the cubic and the quadratic plus linear terms. The uniform measure $\mc P_2(n)$ is invariant under addition of an arbitrary quadratic polynomial. 
For arbitrary distributions over cubic polynomials, just apply the argument above pointwise. 
\end{proof}
In particular, \cref{lem:moments iqp} implies that 
\begin{align}
\label{eq:second third moment iqp}
  m_2 = 3 + o(1), \quad \text{ and } \quad  m_3 = 15 + o(1). 
\end{align}

Finally, we note that for uniformly random cubic polynomials, \textcite{dalzell_how_2020} proved that all moments match those of random functions. We state a slightly more precise version of the theorem.
\begin{lemma}[\cite{dalzell_how_2020}]
\label{thm:uniform iqp moments}
For all $n,k \in \mb N$
\begin{align}
  \left| 2^{nk} \Eb_{f \sim \mc P_3(n)} \ngap(f)^{2k} - (2k-1)!! \right| \le C_k2^{-n}, 
\end{align}
where $C_k  = (2k-1)!! \, 2^k + k \,2^{k^2}$.  
\end{lemma}
\begin{proof}[Proof sketch]
The proof is verbatim that of \cite[Theorem~8]{dalzell_how_2020}, and only needs to make the $k$-dependence explicit. 
The only implicit constant is $c_k'$ appearing in the last paragraph of the proof, the number of subspaces $H \le \FF_2^{2k}$ of dimension $d<k$ satisfying the closure condition $H^\times\subseteq H^\perp$ (see \cite{dalzell_how_2020} for the definitions of $H^\times, H^\perp$), which the original proof only asserts to ``depend on $k$.''
In particular, every $H$ lies in the even-weight subspace of $\FF_2^{2k}$, which has dimension $2k-1$. 
We bound $c_k'$ crudely by the total number of such even-weight subspaces of $\mathbb F_2^{2k}$ of dimension $d<k$, via the Gaussian binomial estimate
\begin{align}
\binom{2k-1}{d}_2 \le 2^{d(2k-d)} \le 2^{k^2-1} \quad\text{for } d<k,
\qquad\text{so}\qquad c_k' \le \sum_{d=0}^{k-1} \binom{2k-1}{d}_2 \le k\cdot 2^{k^2-1}. 
\end{align}
Inserting this into the original error accounting with the $d=k$ correction at most $(2k-1)!!\cdot 2^k 2^{(k-1)n}$ and the $d < k$ terms contributing at most $c_k' 2^{(k-1)n}$ shows the claim.
\end{proof}

\begin{lemma}[Min-entropy bounds for random IQP circuits]
\label{lem:min-entropy}
We have the min-entropy bounds
\begin{align}
  \Pr_{f \sim \mc M(n)}\left[H_\infty(p_f) \ge  2n/3 - O(\log(1/\delta)) \right] &\ge 1- \delta\\
  \Pr_{f \sim \mc P_3(n)}\left[H_\infty(p_f) \ge n - O(\sqrt n) \right] &\ge 1- \exp(-n).
\end{align}
\end{lemma}
\begin{proof}
We show the lemma using a general min-entropy bound based on moments. 
We follow the proof of Lemma~5 in Ref.~\cite{hangleiter_sample_2019} with the second
moment replaced by the $k$-th. 
We then use that  $H_\infty(p) \ge \frac{k-1}{k}H_k(p)$, where 
\begin{align}
  H_\alpha(p) \coloneqq - \frac{1}{\alpha-1} \log ( \sum_{x \in \FF_2^n} p(x)^\alpha )
\end{align}
is the $\alpha$-R\'enyi entropy. 
This gives us that for any $k,n \in \mb N$, given any distribution $\mc D(n)$ over $n$-bit Boolean polynomials, with probability at least $1-\delta$ over the choice of $f \sim \mc D(n)$,
$$
H_\infty\big(p_{f}\big) \ge \frac{1}{k}
\left( \log\delta - \log\, \Eb_{f \sim \mc D(n)}\left[\sum_{x\in\FF_2^n}
p_{f}(x)^{k}\right] \right).
$$
For the first bound, 
we choose $k=3$ and use \cref{lem:moments iqp}, for the second, 
we let $k=\lfloor\sqrt n\rfloor$, $\delta = \exp(-n)$ and use \cref{thm:uniform iqp moments}.
\end{proof}

\subsection{Verification via XEB}
\label{sub:verification_via_xeb}

Our verifiable quantum advantage scheme is based on sampling from the output distribution of IQP circuits randomly drawn from some distribution $\mc M(n)$ over cubic polynomials as in \cref{lem:moments iqp}. 
Our quantum advantage test is based on the so-called \emph{cross-entropy benchmark (XEB)} defined for two distributions $q,p: \FF_2^n \rightarrow [0,1]$ as follows: 
\begin{align}
\label{eq:single xeb}
   \chi(q,p) \coloneqq 2^n \sum_x q(x) p(x) - 1,  
 \end{align} 
where we think of $q$ as the sampled distribution and $p$ as the target distribution. 
The XEB can be estimated empirically, given samples $x_1, \ldots, x_\ell \leftarrow q$ as 
\begin{align}
  \hat \chi(x_1, \ldots, x_\ell) = 2^{n} \cdot \frac 1 \ell  \sum_{i=1}^\ell p(x_i) -1. 
\end{align}
We can further define the average XEB with respect to the measure $\mc M(n)$ as  
\begin{align}
\label{eq:average xeb}
  \overline \chi \coloneqq \Eb_{f \leftarrow \mc M(n)} \chi(q_f,p_f), 
\end{align}
where for every target distribution $p_f$, we consider an actually sampled distribution $q_f$. 
Again, the average XEB can be estimated given samples $x_1^f, \ldots, x_\ell^f \leftarrow q_f$ for $f= f_1, \ldots, f_L$ sampled from $\mc M(n)$ as 
\begin{align}
  \hat {\overline \chi}(\{ x_i^{f_j} \}_{ij}) = {2^{n}}  \cdot \frac 1{\ell L} \sum_{j=1}^L \sum_{i=1}^\ell p_{f_j}(x_i^{f_j} ) -1 . 
\end{align}
It is easy to see that 
\begin{align}
  \overline \chi = \begin{cases}
    0 & \text{ if  }  \forall f:  q_f = \id/2^n \\
    m_2 -1 & \text{ if  }   \forall f: q_f = p_f,
  \end{cases}
\end{align}
and therefore the average XEB serves as a bona-fide test that distinguishes uniform samples from ideal samples with ideal score $\overline \chi = 2 + o(1)$ by \cref{eq:second third moment iqp}. 
In fact, it has been conjectured that this is already sufficient for a quantum advantage test in that no classical polynomial-time algorithm will be able to achieve a high average XEB value for random quantum circuits \cite{aaronson_complexity-theoretic_2017,aaronson_classical_2020}. We formalize this task as follows. 

\begin{task}[$\xeb_{b,\ell}$]
  Given polynomials $f_1, \ldots, f_L \in \FF_2[x_1, \ldots, x_n]$, output samples $x_1^{f_1}, \ldots, x_\ell^{f_L}$ satisfying 
  \begin{align}
    \hat {\overline \chi}(\{ x_i^{f_j} \}_{i \in [\ell], j \in [L]}) \ge b.
  \end{align}
\end{task}

We state here the hardness of achieving high XEB for random cubic polynomials as a conjecture analogous to the arguments of \textcite{aaronson_classical_2020}\footnote{Strictly speaking, \textcite{aaronson_classical_2020} reduce the hardness of \xeb\ to a task they dub XQUATH, which shall not concern us here, since its hardness is a stronger conjecture and has been proven false for certain low-depth settings, where achieving high XEB is conceivably still hard \cite{gao_limitations_2024,aharonov_polynomial-time_2023}.}. 
\begin{conjecture}[Computational soundness of \xeb]
  Let  $f_1, \ldots, f_L$ be drawn independently from $\mc P_3(n)$, and let $L, \ell = \poly(n)$. 
  Then there is no $\poly(n,L,\ell)$-time classical algorithm that solves $\xeb_{1,\ell}(f_1, \ldots, f_L)$ with probability at least $ 1- 1/\poly(L)$ over the choice of polynomials and the algorithm's randomness. 
\end{conjecture} 

Moreover, for completeness of the XEB test, we would like to ensure that a sampler from the ideal distribution $p_f$ indeed passes this test with high probability for sufficiently large $L, \ell$. To see why this is true, consider the estimator obtained for $\ell = 1$. 
\begin{lemma}[Completeness for the ideal sampler]
\label{lem:completeness xeb}
For $\ell =1$ and $q_f = p_f$, $\hat{\overline \chi} \ge 1$ with probability at least $1- 10/L$. 
\end{lemma}
\begin{proof} Let $Z = 2^n p_f(X)-1$ with $f \sim \mc M(n)$ and $X \sim p_f$. Then $\Eb[Z] = \overline \chi$. Using the fact that $\mc M(n)$ is invariant under additions of arbitrary linear terms, which correspond to the different outcomes (the hiding property), 
the single-shot variance is given by 
  \begin{align}
  \Var[Z] = 2^{3n}\Eb_{f \leftarrow \mc M(n)}  p_f(x)^3 - \left(2^{2n} \Eb_{f \leftarrow \mc M(n)} p_f(x)^2\right)^2= m_3 - m_2^2 = 6 + o(1), 
\end{align}
and therefore $\Var[
  \hat {\overline \chi}
] = ( m_3 - m_2^2 )/L$. 
By Chebyshev's inequality (choosing $\epsilon = m_2 -2$) we therefore find 
\begin{align}
   \Pr[\hat{\overline \chi} < 1] \le  \frac{\Var[\hat{\overline \chi}]}{(m_2 - 2)^2}, 
 \end{align} 
 which completes the proof.
\end{proof}

\subsection{Simulating IQP circuits with and without planted independent spaces}
\label{sub:simulating_iqp_circuits_with_and_without_pis}

In this section, we will establish that knowledge of a large (planted) independent space in a polynomial $f$ as in \cref{eq:cubic polynomial normal form} reduces the complexity of (strongly) simulating the corresponding IQP circuit in the sense of computing the outcome amplitudes $a_f$ up to exponentially small error.  
This is the meaningful task to consider, since all known sampling algorithms for random IQP circuits involve computing the outcome amplitudes as a subroutine \cite{hangleiter_computational_2023}.

\paragraph{Simulating random polynomials}
Let us first consider the complexity of simulating IQP circuits corresponding to random cubic polynomials. The best known algorithm has been shown by \textcite{dalzell_how_2020} building on the breakthrough results of \textcite{lokshtanov_beating_2017}. 
Surprisingly, this algorithm achieves a runtime that is asymptotically better than $2^n$. In other words, 
to compute the gap of a Boolean polynomial, one need not na\"ively evaluate the polynomial on all inputs. 
\begin{theorem}[Asymptotic complexity \cite{dalzell_how_2020}]
\label{thm:asymptotic complexity gap}
  There is an algorithm that computes $\gap(f)$ for a cubic polynomial $f \in \FF_2[x_1, \ldots, x_n]$ in time $\tilde O(2^{0.9965n})$. 
\end{theorem}
The central idea underlying this algorithm is the observation that the partial sum 
\begin{align}
  F_b(y) = \sum_{x \in \FF_2^b} (-1)^{f(y,x)},
\end{align}
can be expressed as a low-degree polynomial modulo a certain power of $2$. 
The running time depends on the cost of constructing this modular
representation and recovering the partial sums from it.
The polynomial for $b = cn$ for some $c > 0$ can then be evaluated in time strictly less than $2^{(1-c)n}$ and therefore the overall sum $\sum_y F_b(y)$ can be evaluated in time less than $2^n$. 
The fundamental limit of this variable-splitting approach appears to be runtime $\tilde O(2^{n/2})$ since the best one can hope for is to choose $b = n/2$, i.e., equipartition the variables, and find an algorithm that computes $F_{n/2}(y)$ for each $y$ in time at most $\tilde O(2^{n/2})$ \cite{dalzell_how_2020}.
This limit is achieved by an \ma\ algorithm for computing the gap due to \cite{williams_strong_2016}.

In upcoming work, we build on this idea and develop an algorithm with a significantly improved exponent, approaching the limit of this approach at $n/2$. 
\begin{theorem}[Improved gap algorithms \cite{upcoming}]
\label{thm:upcoming gap theorem}
  Given a cubic polynomial $f(x) \in \FF_2[x_1, \ldots, x_n]$, $\gap(f)$ can be computed exactly in expected time $2^{0.552n + O(\sqrt n)}$. 
\end{theorem}

The counting version of the strong exponential time hypothesis (\#SETH) implies that for any $\epsilon >0$ there is a constant $d >2$ such that the gap of degree-$d$ Boolean polynomials cannot be computed in time $(2-\epsilon)^n$ \cite{williams_counting_2018} and the same statement holds true for any field $\FF_q$ \cite{dell_solving_2024}. 
While these results do not rule out a specific exponent for $d=3$, they rule out algorithms that are not sensitive to the polynomial degree. 
For these reasons, we believe that even achieving runtime $\tilde O(2^{n/2})$ without a space constraint would constitute a breakthrough.

Contrasting with the standard approach of evaluating the sum over $2^n$ phases, this asymptotically improved approach has the drawback that it needs to store explicit coefficient tables for $F_b$ in the middle of the computation. These will be exponentially large, and in fact for the algorithm of \cref{thm:upcoming gap theorem}, they have size $2^{0.552n + O(\sqrt n)}$, matching the time bound.

However, in practice, space is a more limited resource than time, in particular so for parallelizable algorithms such as the ones used for computing the gap. 
This motivates considering only algorithms that use polynomial space. 
To the best of our knowledge, the best such algorithm has runtime $O(n^2 2^n)$ compared to the na\"ive $O(n^32^n)$. 
This algorithm is based on using Gray codes to efficiently evaluate the gap. 
Moreover, there is no known algorithm that achieves high XEB values in the absence of noise that circumvents full-circuit simulation. We therefore conjecture that in fact achieving a high constant XEB score requires genuinely exponential time for some constant $\kappa > 1/2$. 

\begin{conjecture}[Exponential XEB (\Exp\xeb) conjecture]
\label{conj:exp xeb full}
  There exists a $\kappa > 1/2$ such that any classical polynomial-space algorithm which, given uniformly random $f_1, \ldots, f_L \leftarrow \mc P_3(n)$ solves $\xeb_{1,\ell}(f_1, \ldots, f_L)$ for $\ell, L = \poly(n)$ with at most inverse polynomial failure probability, requires time $\Omega(2^{\kappa n} \cdot \poly(n))$. 
\end{conjecture}

\paragraph{Simulating planted polynomials}

The runtime reduction for planted polynomials is possible due to an observation of \textcite{maslov_fast_2024}. 
\begin{definition}[Vertex cover]
  Let $H = (V,E)$ be a hypergraph. Then a \emph{vertex cover} is a set $C \subset V$ satisfying 
  $\forall e \in E \, \exists v \in e: v \in C$,
  i.e., every hyperedge overlaps with $C$. 
\end{definition}

The key observation is that given $f$, a vertex cover of the hypergraph $H(f)$ defined in \cref{def:poly hypergraph} yields a low-rank stabilizer decomposition of $a_f$. For notation, given a set $S \subset[n]$ and a bit string $x \in \FF_2^n$ write $x_S = (x_i)_{i \in S} \in \FF_2^{|S|}$ for the substring corresponding to the indices in $S$, and let $S^c = [n]\setminus S$. 

\begin{lemma}[Vertex cover decomposition of amplitudes]
\label{lem:vertex cover amplitudes}
  Let $f \in \FF_2[x_1, \ldots, x_n]$ be a cubic polynomial. Given a vertex cover $C \subset [n]$ of $H(f)$, we can write 
  \begin{align}
    a_f(x) = \frac 1 {2^{|C|}}\sum_{z \in \FF_2^{|C|}} \ngap(f^C_{z} + \langle x_{C^c}, \cdot \rangle ) (-1)^{\langle x_C, z\rangle}, 
  \end{align}
  for quadratic polynomials $f^C_z \in \FF_2[x_1, \ldots,x_{|C^c|}]$ specified in the proof. 
\end{lemma}
\begin{proof}
For simplicity, let us assume that $f$ is homogeneous. 
Then we can write 
\begin{multline}
  f(x)= \sum_{j < k \in C^c} \left( \sum_{i \in C} (F_{ijk}  + F_{jik} + F_{jki} ) x_i  \right) x_j x_k + \sum_{k \in C^c} \left( \sum_{i< j \in C} (F_{ijk}  + F_{ikj} + F_{kij} )x_i x_j  \right) x_k  \\+ \left(\sum_{i< j < k \in C} F_{ijk}x_i x_j x_k\right)  \eqqcolon f^C_{x_C}(x_{C^c}),
\end{multline}
which defines an (inhomogeneous) quadratic polynomial over $\FF_2$, given an assignment of the cover variables $x_i, i \in C$. 
This means that we can write the outcome amplitudes as 
\begin{align}
  2^n a_f(x) & =  \sum_{y_C} (-1)^{\langle x_{C},y_{C}\rangle} \left(  \sum_{y_{C^c}} (-1)^{f^C_{y_C}(y_{C^c}) + \langle x_{C^c},y_{C^c}\rangle} \right) \\
  & = \sum_{z \in \FF_2^{|C|}} \gap(f^C_{z} + \langle x_{C^c}, \cdot \rangle ) (-1)^{\langle x_C, z\rangle}. 
\end{align}
The inhomogeneous case follows from the fact that the remaining terms are quadratic to begin with. 
\end{proof}
The vertex-cover decomposition directly implies an algorithm for strongly simulating---i.e., computing the amplitudes of---cubic IQP circuits. 
\begin{theorem}
\label{thm:vertex cover algorithm}
  Let $f$ be a cubic polynomial with $h$-dimensional independent space $V$. There is an algorithm that exactly computes $a_f(x)$ in time $O(2^{n-h} n^3)$ given $f$ and a basis of $V$. 
\end{theorem}
\begin{proof}
To start, let $V = \langle e_1, \ldots, e_{h} \rangle $ be the subspace spanned by the first $h$ coordinates. Then $C = [h+1, \ldots n]$ is a vertex cover. 
The proof follows directly from applying the polynomial-time algorithm for computing the gap of quadratic polynomials \cite{ehrenfeucht_computational_1990} to the decomposition of  \cref{lem:vertex cover amplitudes}. 

For an arbitrary subspace, let $A \in \gl(n, \FF_2)$ be the change of basis that maps the first $h$ coordinates to $V$, i.e., $V = A(\langle e_1, \ldots, e_h \rangle )$. Then $a_f(x) = a_{f\circ A}(A^Tx)$ and $f \circ A$ has vertex cover $[h+1, \ldots, n]$, so we can apply the algorithm above. 
\end{proof}

\begin{remark}[Quantum estimation of $\gap$]
The gap of cubic polynomials can be estimated on a quantum computer in time $O(2^{n/2})$ for typical instances. To see this, let $\ket {\psi_f} = H^{\otimes n } C(f) H^{\otimes n} \ket{0^n} $ be the output state of the IQP circuit. 
Now, observe that applying phase estimation on the Grover operator $(\id - 2\proj 0)(\id - 2 \proj{\psi_f})$ can estimate the angle $\theta = 2 \arcsin(\ngap(f))$ up to precision $\epsilon$ in time $O(1/\epsilon)$  \cite{brassard_quantum_2002}. But  $\ngap(f) \sim 2^{-n/2}$ for typical $f$ which shows the claim. 
\end{remark}

\paragraph{Proof of main result}
We are now ready to prove our main result (\cref{thm:main result}), which we here state formally. 
\begin{theorem}
  [Verifiable advantage from planted independent spaces]
\label{thm:main result full} 
  Assume the \Exp\PIS\ and \Exp\xeb\ conjectures (\ref{conj:PIS} and \ref{conj:exp xeb full}). Let $h = (1-c)n$ with $1/2 < c < \kappa \le 1 $ and $\ell, L = \poly(n)$. 
  The verifier chooses $f_1, \ldots, f_L \leftarrow \mc P_3(n,h)$ independently and asks the prover to solve $\xeb_{1, \ell}(f_1, \ldots, f_L)$. Then, 
  \begin{enumerate}[label=(\roman*)]
    \item given samples $\{x_i^{f_j}\}_{i \in [\ell], j \in [L]}$ from the prover, using knowledge of the planted independent spaces, the verifier can compute $\hat{\overline \chi}(\{x_i^{f_j}\}_{i \in [\ell], j \in [L]})$ in time $O(L\ell  n^3 2^{cn})$ and polynomial space, 
    \item an honest, polynomial-time quantum prover passes $\xeb_{1, \ell}(f_1, \ldots, f_L)$ test with probability at least $1- 10/L$, 
    \item there is no classical polynomial-space algorithm with runtime $o(2^{\kappa n})$ that passes the XEB test with failure probability $1/\poly(n)$. 
  \end{enumerate}
\end{theorem}

\begin{proof}
(i) follows from \cref{thm:vertex cover algorithm} and (ii) follows from \cref{lem:completeness xeb}. To show (iii), consider the following hybrid conjecture. 
\begin{conjecture}
\label{conj:hybrid}
  There exists a $1/2 < \kappa \le 1$ such that any classical polynomial-space algorithm which, given random planted polynomials $f_1, \ldots, f_L \leftarrow \mc P_3(n,h)$ solves $\xeb_{1,\ell}(f_1, \ldots, f_L)$ for $\ell, L = \poly(n)$ with at most inverse polynomial failure probability, requires time $\Omega(2^{\kappa n} \cdot \poly(n))$. 
\end{conjecture}
\cref{conj:hybrid} directly implies (iii). We now argue that it is implied by \cref{conj:exp xeb full} and \ref{conj:PIS}. 
To this end, suppose \cref{conj:hybrid} was false. Then there exists an algorithm $\mc A$ with runtime $o(2^{\kappa n})$ using polynomial space which passes $\xeb_{1,\ell}$ with inverse polynomial failure probability for planted polynomials. 
We can use $\mc A$ as a distinguisher between $\mc P_3(n)$  and $\mc P_3(n,h)$ by running $\mc A$ on random  $f_1, \ldots, f_L \leftarrow  \mc P_3(n,h)$ and outputting \texttt{planted} if it terminates in time $o(2^{\kappa n})$ and the returned samples pass $\xeb_{1, \ell}$. 
On the other hand, \cref{conj:exp xeb full} implies that for uniformly random polynomials, any polynomial-space algorithm passing $\xeb_{1, \ell}$ requires runtime $\Omega(2^{\kappa n})$. 
In either case, the XEB can be computed in time $O(2^{0.9965 n})$ by \cref{thm:asymptotic complexity gap}. 
Therefore $\mc A$ yields a distinguisher with overall runtime $o(2^{n})$, contradicting \cref{conj:PIS}. 
\end{proof}

\subsection{Towards certified randomness from planted IQP circuits}
\label{sec:certified randomness}

In this section, we sketch the application of IQP circuit sampling with planted independent spaces to generate classically verifiable randomness. 
The idea is to apply the framework of \textcite{aaronson_certified_2023} for verifiable randomness from random circuit sampling to our setting of planted IQP circuits.
This is possible since as shown in \cref{lem:min-entropy}, the output distributions of random IQP circuits---like random circuits---have high min-entropy with high probability. 
But in contrast to random circuits, there are instance sizes of planted IQP circuits that are verifiable but not classically simulatable. 
This overcomes the key obstacles present in the protocol by \textcite{aaronson_certified_2023,liu_certified_2025}: Their protocol requires that a quantum server quickly respond to queries in order to prevent it from spending classical resources on spoofing XEB, while the classical client is allowed to spend a long time verifying the obtained samples post-hoc. 
In contrast, when basing the certifiable randomness protocol on planted IQP circuits, 
we showed that \Exp\xeb\ and \Exp\PIS\ imply a gap between verification and simulation time, allowing us to certify samples more efficiently. This means that the timing requirement can at least be relaxed, if not removed from the protocol since spoofing  the protocol using a classical computer will be out of reach for any available resources.

Crucially, however, certifiable randomness requires stronger guarantees on the hardness of the solved task than quantum advantage, since we must ensure that no  \emph{efficient quantum adversary} can pass the protocol using samples with low entropy. 
We conjecture that this is true. 
While a na\"ive application of the \cite{aaronson_certified_2023} protocol concerns polynomial-time quantum provers, we discuss stronger quantum hardness conjectures and attacks with a characteristic exponent $1/2$. 
We defer a detailed analysis of the implications of $\tilde O(2^{n-h})$-time classical and $O(2^{n/2})$-time quantum spoofing algorithms for the security of the protocol to future work.

\subsubsection{The certified randomness protocol}

We consider the protocol of \textcite{aaronson_certified_2023} for random cubic IQP circuits with planted independent spaces. 
We showed in \cref{lem:min-entropy} that both planted cubic IQP circuits and uniformly random cubic IQP circuits generate $2n/3 - O(\log 1/\delta)$ and $n - O(\sqrt n)$ bits of min entropy with high and overwhelming probability, respectively.
The idea of Aaronson and Hung is that---under their hardness assumption---any efficient quantum adversary that passes \xeb\ (for random circuits, in their case) with high probability must have generated samples with a linear amount of average entropy. 

For the purpose of this work, we restrict ourselves to their 
single-round entropy analysis, without adversarial side information.
Their reduction turns a sufficiently successful low-entropy prover into a protocol for the following task.

\begin{task}[Long List Quantum Supremacy Verification $\llqsv(\mc U)$ \cite{aaronson_certified_2023}]
  We are given oracle access to $M=2^{3n}$ quantum circuits $C_1,\ldots,C_M$, each on $n$ qubits, which are promised to be drawn independently from the distribution $\mc U$. We're also given oracle access to $M$ strings $s_1,\ldots,s_M\in\{0,1\}^n$. Then the task is to distinguish the following two cases:
  \begin{enumerate}[label=(\arabic*)]
    \item \textbf{No-Case}: Each $s_i$ is sampled independently, uniformly from $\{0,1\}^n$. 
    \item \textbf{Yes-Case}: Each $s_i$ is sampled independently from $p_{C_i}$, the output distribution of $C_i$.
  \end{enumerate}
\end{task}
They make the following conjecture. \begin{conjecture}[$\llha_B(\mc U)$]
\label{conj:llha}
  For some parameter $B<n$:
  \begin{align}
    \llqsv(\mc U)\notin\qcam\Time(2^Bn^{O(1)})/\textsf{q}(2^Bn^{O(1)}),
  \end{align}
  where \qcam\ is the class of problems that admit an \am\ protocol with classical communication and a quantum verifier, and $\qcam\Time(T)$ is a generalization where the verifier can use running time $T$ but communication is still polynomial in $n$. Finally, in $\qcam\Time(T)/\textsf{q}(A)$ we additionally give the verifier $A$ qubits of advice that depend only on $n$.
\end{conjecture}

Based on this conjecture, they show the following theorem.
\begin{theorem}[{\cite[Theorem~5.10]{aaronson_certified_2023}}]
\label{thm:single-round-entropy}
Let $\mc U$ be an ensemble of $n$-qubit circuits satisfying
$\llha_B(\mc U)$, and let $b>0$ be fixed.
Let $\mc A$ be a quantum polynomial-time algorithm that solves $\xeb_{b, \ell}(C)$ for $C \leftarrow \mc U$ with probability at least $q$ over the choice of circuit and $\mc A$'s internal randomness. Then $\mc A$'s output $s_1, \ldots, s_\ell$ satisfies
\[
  H(s_1, \ldots, s_\ell \mid C)\ge
  \frac B2\left(\frac{(b+1)q-1}{b}-o(1)\right),
\]
where $H$ denotes the average Shannon entropy $H(\, \cdot\, | C ) = \Eb_{c \leftarrow \mc U}[H(\, \cdot \,| C = c )]$.
\end{theorem}

The key idea of the reduction is that there is an \am\ protocol which uses an adversary who passes \xeb\ with too little average entropy in order to solve \llqsv. 
Intuitively, the assumption (\cref{conj:llha}) is plausible because the output distributions of the considered circuit ensemble have very high min entropy.

\subsubsection{Application to cubic IQP circuits}

We can straightforwardly apply the single-round result of \textcite{aaronson_certified_2023} to the uniformly random cubic IQP ensemble under $\llha_B(\mc P_3(n))$. 
For that ensemble, the verifier has no planted-space simulation secret and the protocol is completely analogous to the random-circuit one, where verification can be performed after receiving a fast response. 
We want to apply the certified randomness argument to instances with planted independent spaces that ease verification. 
To this end, we assume $\llha_B(\mc P_3(n,h))$ with $B=\Theta(n)$ and apply Aaronson and Hung's reduction to this planted ensemble.\footnote{Here, in some abuse of notation, we denote the distribution over circuits $H^{\otimes n } C(f) H^{\otimes n}$ for $f \leftarrow \mc P_3(n,h)$ by $\mc P_3(n,h)$ as well.} 
This will rule out quantum polynomial-time algorithms that pass $\xeb_{b,\ell}$ without producing high-entropy samples. A simple result we can directly infer is an entropy bound obtained against polynomial-time quantum provers. 

\begin{corollary}[Entropy bound for planted IQP circuits]
\label{cor:planted-randomness}
Assume $\llha_B(\mc P_3(n,h))$ with $B=\Theta(n)$.
Let $\mc Q$ be a polynomial-time quantum prover that, given
$f\leftarrow\mc P_3(n,h)$, passes the single-circuit test
$\xeb_{1,\ell}(f)$ with probability at least $1-\varepsilon$.
Then $\mc Q$'s output $s_1,\ldots,s_\ell$ satisfies
\[
  H(s_1, \ldots, s_\ell \mid f)\ge
  \frac B2\left(1-2\varepsilon-o(1)\right).
\]
\end{corollary}

\begin{proof}
Apply \cref{thm:single-round-entropy} to the planted IQP
circuit ensemble with $b=1$ and $q\ge1-\varepsilon$. 
\end{proof}
For a full security analysis of our protocol, we also want to rule out classical polynomial-space provers with a time budget comparable to the verifier's $2^{n-h}\poly(n)$.
This requires a separate classical hardness assumption for the same planted ensemble, since  the polynomial-time restriction in \cref{thm:single-round-entropy} does not cover exponential-time classical provers.
Specifically, for $h=\lfloor\alpha n\rfloor$ with fixed $0<\alpha<1/2$, assume there exist constants $\kappa>1-\alpha$ and $\varepsilon_{\rm cl}\in(0,1/4)$ such that every randomized classical polynomial-space algorithm of runtime $o(2^{\kappa n})$ passes this test with probability at most $1-\varepsilon_{\rm cl}$ for all sufficiently large $n$.
Such algorithms cannot attain the acceptance probability in \cref{cor:planted-randomness} for fixed $\varepsilon<\varepsilon_{\rm cl}$.
Deriving this assumption from \Exp\xeb\ and \Exp\PIS\ requires a reduction that preserves the test's success parameters and fits the assumed time and space bounds.
We leave these details of a complete protocol with the appropriate security guarantees for future work.

\section{Worst-case hardness of the independent space problem}
\label{sec:worst-case hardness}

For the moment, we are working over arbitrary fields $\FF$ in this section. In this section, we will prove that it is $\NP$-hard in the worst-case to decide whether a given cubic form has an independent subspace of a given size. This is done by showing that a cubic form has an independent subspace of large size iff its symmetric slice rank is low. We then show that it is $\NP$-hard to decide whether the symmetric slice rank of a cubic form is bounded. \textcite{brand2026partitionrankalgebraiccircuit} have recently shown the same hardness result for symmetric slice rank with a different proof, independently of our work. We still keep our proof here, since our main result is that for a trilinear form whose support is an antichain, slice rank as a trilinear form and symmetric slice rank as a cubic form coincide (Corollary~\ref{cor:symSR-SR}). This is a structural insight of independent interest.

\subsection{Slice rank (trilinear case)}

Let $x$, $y$, and $z$ be three sets of variables.
\begin{definition}[\textcite{taoblogSR}]
Let $P \in \FF[x,y,z]$ be a trilinear form.
\begin{enumerate}
    \item A slice decomposition is a collection of linear forms $\ell_{x,i}(x)$ and bilinear forms $B_{x,i}(y,z)$, $1 \le i \le r_x$,  linear forms $\ell_{y,j}(y)$ and bilinear forms $B_{y,j}(x,z)$, $1 \le j \le r_y$, as well as  linear forms $\ell_{z,k}(z)$ and bilinear forms $B_{z,k}(x,y)$, $1 \le k \le r_z$, such that
    \[
       P(x,y,z) = \sum_{i = 1}^{r_x} \ell_{x,i}(x) B_{x,i}(y,z) + \sum_{j = 1}^{r_y} \ell_{y,j}(y) B_{y,j}(x,z) + \sum_{k = 1}^{r_z} \ell_{z,k}(z) B_{z,k}(x,y).
    \]
    \item The slice rank $\SR(P)$ is the minimum $r$ such that there is a slice decomposition as above with $r_x + r_y + r_z \le r$.
\end{enumerate}
\end{definition}

\begin{theorem}[\textcite{BlaserILPS21}]
    Deciding whether the slice rank of a given trilinear form is smaller than a given $r$ is $\np$-hard.
\end{theorem}

\begin{definition}
A trilinear form $P(x,y,z) = \sum_{i = 1}^\ell \sum_{j = 1}^m \sum_{k = 1}^n  p_{i,j,k} x_i y_j z_k$ has an independent triple of spaces of dimension $h$ if there are spaces $V_x \le \FF^\ell$, $V_y \le \FF^m$, and $V_z \le \FF^n$ of dimensions $h_x$, $h_y$, and $h_z$ with $h_x + h_y + h_z = h$ such that $P(V_x,V_y,V_z) = \{0\}$.
\end{definition}

If there is an independent triple as above, then there are $A \in \gl(\ell, \FF)$, $B \in \gl(m, \FF)$, and $C \in \gl(n, \FF)$ such that in $P(Ax,By,Cz) = \sum_{i = 1}^\ell \sum_{j = 1}^m \sum_{k = 1}^n  p'_{i,j,k} x_i y_j z_k$, $p'_{i,j,k} = 0$ for all $1 \le i \le h_x$, $1 \le j \le h_y$, and $1 \le k \le h_z$.

\begin{remark}
The notion of an independent triple is less intuitive for trilinear forms than independent spaces for cubic forms, as there is always an independent triple of dimension $m + n$ by setting $h_x = 0$ or any of the other two parameters. But this is in accordance with the fact that the slice rank is at most $\ell$ and this trivial triple corresponds to the slice decomposition in which all linear forms depend only on $x$.
\end{remark}

There is a one-to-one correspondence between the slice rank of a trilinear form and the dimension of an independent triple of spaces, which immediately follows from \cite[Lemma~1(iv)]{taoblogSR}.

\begin{theorem}
For any trilinear form $P$, $\SR(P) \le r$ iff there is an independent triple of spaces of dimension $\ge \ell + m + n - r$.
\end{theorem}

For a trilinear form $P(x,y,z) = \sum_{i = 1}^\ell \sum_{j = 1}^m \sum_{k = 1}^n p_{i,j,k} x_i y_j z_k$, let $\supp(P) = \{(x_i,y_j,z_k) \mid p_{i,j,k} \not= 0 \}$. $(x \times y \times z, \supp(P))$ is a 3-uniform, 3-partite hypergraph, which, by a slight abuse of notation, we also denote by $\supp(P)$. 
For a hypergraph $H$, let $\tau(H)$ denote the vertex cover number, that is, the minimum size of a vertex cover.

\begin{lemma}[\cite{taoblogSR}]
    $\SR(P) \le \tau(\supp(P))$.
\end{lemma}

Since the slice rank is invariant under invertible linear transformations of the variables, that is, under $\gl(\ell, \FF) \times \gl(m, \FF) \times \gl(n, \FF)$, we can take the minimum over all supports under this action on the right hand side.

Assume that there are three total orderings on the variable sets $x$, $y$, and $z$, respectively. Consider the induced (partial) product order on $x \times y \times z$ defined by
\[
(x,y,z)\ge (u,v,w)
\quad\Longleftrightarrow\quad
x\ge u,\quad y\ge v,\quad z\ge w.
\]
We moreover say $(x,y,z) > (u,v,w)$ if $(x,y,z) \ge (u,v,w)$ and $(x,y,z) \neq (u,v,w)$. 
Let $\max \supp(P)$ denote all triples in the support that are maximal in the product order.

\begin{lemma}[\cite{taoblogSR}]
   $\SR(P) \ge \tau(\max \supp(P))$. 
\end{lemma}

If $\supp P$ is an antichain, that is, all elements are maximal, then the slice rank of $P$ equals $\tau(\supp P)$.

\begin{definition}
A trilinear form $P(x,y,z)$ is called \emph{oblique} if there are linear transforms $A$, $B$, $C$ such that $\supp P(Ax,By,Cz)$ is an antichain.
\end{definition}

The hardness proof by \textcite{BlaserILPS21} only constructs instances of oblique trilinear forms. Even stronger, all instances are presented in such a way that the support of the constructed trilinear form is already an antichain.

\subsection{Symmetric slice rank (polynomial case)}

Let $P(x)$ be a cubic form and let us write $P$ in the standard form 
\begin{align}
   P(x) = \sum_{i \le j \le k} p_{i,j,k} x_i x_j x_k.
 \end{align} 
Over the binary field, the $=$ cases capture lower-degree terms since $x_i^2 = x_i$.

\begin{definition}
Let $P \in \FF[x]$ be a cubic form. 
\begin{itemize}
    \item A symmetric slice decomposition is a collection of linear forms $\ell_i(x)$ and quadratic forms $Q_i(x)$, $1 \le i \le r$ such that
\[
    P(x) = \sum_{i = 1}^r \ell_i(x) Q_i(x)
\]
    \item $P$ has symmetric slice rank $\le r$, denoted by $\symSR(P) \le r$, if there is a symmetric slice decomposition with $\le r$ terms. 
\end{itemize}
\end{definition}

\begin{remark}
We have $\symSR(P) = \symSR(P \circ A)$ for any $ A \in \gl(n, \FF)$ as formal polynomials. Functionally, however, this might not be true over $\FF_2$, the field in which we are most interested. Consider the elementary symmetric polynomial
\[
   E = abc + abd + acd + bcd.
\]
It is easy to check that $\symSR(E) = 2$. However, if we transform $d \to a+b+c+d$ and keep all other variables fixed and reduce all resulting squares, then we get 
\[
   E' = d(ab + ac + bc),
\]
which has slice rank 1. However, our $\NP$-hardness proof also works over $\FF_2$, since we consider trilinear forms whose support forms an antichain and we can work with vertex cover numbers, see Lemma~\ref{lem:symSR-tau}.
\end{remark}

The relation of symmetric slice rank to regular slice rank seems to be rather unexplored.
\textcite{OnetoVenturaRanks} show: The monomial $x_1 x_2^2$ obviously has symmetric slice rank $1$, but the corresponding symmetric trilinear form $\frac 13(x_1 y_2 z_2 + x_2 y_1 z_2 + x_2 y_2 z_1)$ has slice rank $2$. 

We will here do the opposite. We start with a trilinear form $T(x,y,z)$ and view it as a cubic form $P_T(x \cup y \cup z)$ in the variable set $x \cup y \cup z$. In the former case, the natural action on $T$ is $\gl(n_x, \FF) \times \gl(n_y, \FF) \times \gl(n_z, \FF)$, while on $P_T$, it is $\gl(n_x + n_y + n_z, \FF)$. If we order the variables in $x \cup y \cup z$ as $x_1,\dots,x_{n_x}, y_1, \dots,y_{n_y},z_1,\dots,z_{n_z}$, then the support $\supp(T)$ of $T$ and the hypergraph $H(P)$ of $P$ are the same (up to renaming of vertices). 
As a main result we show that for certain trilinear forms $T(x,y,z)$, the slice rank of $T$ and the symmetric slice rank of $P_T(x \cup y \cup z)$ coincide. 

\textcite{ananyan2020small} define a related concept to symmetric slice rank called the \emph{strength}. The strength $\Str(P)$ of a form $P$ of degree $d$ is the minimum $r$ such that there exist $F_i, G_i$, $1 \le i \le r$,  with $1 \le \deg F_i, \deg G_i < d$ such that 
\[
   P(x) = \sum_{i = 1}^r F_i(x) G_i(x).
\]

\begin{remark}
For cubic forms $P$, $\symSR(P) = \Str(P)$.    
\end{remark}

\begin{theorem}
For any cubic form $P$, $\symSR(P) \le r$ iff $P$ has a formally independent space of dimension $\ge n-r$. 
\end{theorem}

Here formally independent means that $P$ is the zero polynomial after plugging the equations of the subspace into $P$. Since formally independent implies functionally independent as defined in Definition~\ref{def:ind_sp}, the ``$\Rightarrow$''-direction is also true functionally. However, we do not need the above theorem for the hardness proof, but will work with \cref{lem:symSR-tau} instead.

\begin{proof}
Let 
\[
    P(x) = \sum_{i = 1}^r \ell_i(x) Q_i(x)
\]
be a symmetric slice decomposition. We can assume that the $\ell_i$ are linearly independent. The kernel of $\ell_1,\dots,\ell_r$ is an independent space and has dimension $n-r$.

For the other direction, assume that $P$ has an independent space $V$ of dimension $n-r$. 
Let $A\in\gl(n,\FF)$ be an invertible matrix whose first
$n-r$ columns form a basis of $V$. 
Then all monomials $x_i x_ j x_ k$ in $P(Ax)$ have an index 
$> n-r$. Thus we can write   
\[
    P(Ax) = \sum_{i = 1}^r x_{n-i+1} Q_i(Ax).
\]
for appropriate quadratic forms. Therefore, $\symSR(P) \le r$.
\end{proof}

We now prove the main lemma of this section, which, given a trilinear form $T$, connects the symmetric rank of the corresponding cubic form $P_T$ to the vertex cover number of $H(P_T)$. 

\begin{lemma} \label{lem:symSR-tau}
    Let $T(x,y,z)$ be a trilinear form. Assume that there is a total ordering on the variables such that every $x$-variable is larger than every $y$-variable and every $y$-variable is larger than every $z$-variable. Furthermore, assume that $\supp T$ is an antichain with respect to the product ordering induced by the total ordering restricted to the three variable sets. Then $\symSR(P_T) \ge \tau(H(P_T))$. 
\end{lemma}

\begin{proof}
Note that $\supp T$ and $H(P_T)$ are isomorphic as hypergraphs. For simplicity, we will simply write $P$ instead of $P_T$ during the rest of the proof.

Let the variables be $x_1,\dots,x_{n_x}$, $y_1,\dots,y_{n_y}$ and 
$z_1,\dots,z_{n_z}$ and assume that this is the ordering, that is, $x_1$ is the largest element, etc. Let $N = n_x + n_y + n_z$.
Let 
\[
   P(x,y,z) = \sum_{i = 1}^r \ell_i(x,y,z) Q_i(x,y,z)    
\]
be a symmetric slice decomposition with linear forms $\ell_1,\dots,\ell_r$. Let $u_1,\dots u_{N-r}$ be a basis of the kernel of the span of $\ell_1,\dots,\ell_r$. Think of $u$ as row vectors and order the entries of $u$ with respect to the total ordering from left to right. That is, the entries $1,\dots,n_x$ correspond to the $x$-variables, the entries $n_x + 1,\dots, n_x + n_y$ to the $y$-variables, and the entries $n_x + n_y + 1,\dots,n_x + n_y + n_z$ to the $z$-variables. Write $u_1,\dots,u_{N-r}$ as a matrix and bring it to row echelon form, that is, each $u_i$ looks like $(0,\dots,0,1,*,\dots,*)$. This can be achieved purely by invertible row operations which just means that we choose a different basis. Assume that $u_i$ has the first $1$ in column $j_i$. Let $D = \{j_1,\dots,j_{N-r}\}$. By the properties of row echelon form, $|D| = N - r$. Let $C = \{1,\dots, N\} \setminus D$. We claim that $C$ is a vertex cover of $H(P)$.

Let $e \in H(P)$ be a hyperedge. Assume that $e$ corresponds to the monomial $c_e \cdot x_a y_b z_c$. Assume for a contradiction that $e$ is not covered, that is, the three indices $a, n_x + b, n_x + n_y + c \in D$, that is, they correspond to a $1$ in the row echelon form. Let $v_a, v_b, v_c$ be the corresponding vectors (each of them is one of the $u_1,\dots,u_{N-r}$). Consider the three-dimensional space $V$ generated by $d_a v_a + d_b v_b + d_c v_c$, where $d_a,d_b,d_c$ are variables. Since $u_1,\dots,u_{N-r}$ is a basis of the kernel of $\langle \ell_1,\dots,\ell_r \rangle$, $P$ vanishes on $V$.   

We now study the behavior of the monomials of $P$ on $V$. 
The entry of an element $v  \in V$ at the position corresponding to $x_a$ in $V$ is $v_{x_a} = d_a$, at the position $y_b$ it is $\lambda d_a + d_b$ for some fixed $\lambda \in \FF$ independent of $v$ and at the position $z_c$ it is $\lambda' d_a + \lambda'' d_b + d_c$ for $\lambda', \lambda '' \in \FF$. This is due to the fact that the indices $a, n_x + b, n_x + n_y + c$ correspond to $1$'s in the row echelon form. Thus, the monomial corresponding to $e$ has the value
\[
   c_e \cdot d_a (\lambda d_a + d_b)(\lambda' d_a + \lambda'' d_b + d_c)
       = c_e \cdot d_a d_b d_c + O(d_a^3 + d_a^2 d_b + d_a^2d_c + d_a d_b^2)
\]
on $V$. 

We will now derive a contradiction to $P$ vanishing on $V$ by showing that the term $d_a d_b d_c$ in the polynomial $P$ on $V$ cannot be generated by any other monomial of $P$ and therefore the term $d_a d_b d_c$ cannot be canceled by any other monomial. 
To this end, let $e' =(a',b',c') \in \supp P$ be an arbitrary different element corresponding to the monomial $x_{a'} y_{b'} z_{c'}$ and consider its behavior on $V$.
We will show that the resulting polynomial does not contain the monomial $d_a d_b d_c$. 

We do this by using that $\supp P$ forms an antichain, 
and that $e'\ne e$, that is, at least one of $x_{a'}>x_a$, $y_{b'}>y_b$, or $z_{c'}>z_c$ must be true.

Case 1: $x_{a'} > x_a$. Then the value at the position corresponding to $x_{a'}$ of any vector in $V$ is $0$. Thus, the monomial corresponding to $e'$ evaluates to $0$ on $V$.

Case 2: $y_{b'} > y_b$. Then the value at the position corresponding to $x_{a'}$ is $\mu d_a$ and the value at the position corresponding to $y_{b'}$ is $\mu' d_{a}$ for some $\mu,\mu' \in \FF$.   Thus the monomial on $V$ takes the form $O(d_a^2(d_a + d_b + d_c))$  

Case 3: $z_{c'} > z_c$. Then the value at the position corresponding to $x_{a'}$ is $\mu d_a$, to $y_{b'}$ is $\mu' d_a + \mu'' d_b$, to $z_{c'}$ is $\mu''' d_a + \mu'''' d_b$ for $\mu, \mu', \mu'', \mu''', \mu '''' \in \FF$. Thus the monomial $x_{a'}y_{b'}z_{c'}$ on $V$ takes the form $O(d_a^3 + d_a^2 d_b + d_a d_b^2)$.

Altogether, we have
\[
  0 = P(V) = c_e \cdot d_a d_b d_c + \text{other terms} 
\]
which implies $c_e = 0$ contradicting the fact that $e \in H(P)$. Therefore, $C$ is a vertex cover of size $r$.
\end{proof}

\begin{corollary} \label{cor:symSR-SR}
If $T(x,y,z)$ is a trilinear form such that $\supp T$ is an antichain, then 
$\SR(T) = \symSR(P_T)$.
\end{corollary}

\begin{proof}
This follows from the fact that $\SR(T) = \tau(\supp T)$ when $\supp T$ is an antichain \cite[Prop.~4]{taoblogSR} and the fact that $\tau(H(P_T)) \le \symSR(P_T)$ (Lemma~\ref{lem:symSR-tau}). We have $\symSR(P_T) \le \SR(T)$, since every slice decomposition of $T$ is also a symmetric slice decomposition of $P_T$. Finally, $H(P_T)$ and $\supp P$ are isomorphic and their cover numbers are the same.
\end{proof}

\begin{corollary}
If $T$ is a trilinear form such that $\supp T$ is an antichain, then $\Str(P_T) = \SR(T)$.    
\end{corollary}

\begin{corollary}
Computing $\symSR(P)$ is $\NP$-hard, that is, deciding whether a given cubic form $P$ has an independent space of dimension $h$ is $\NP$-hard.  
\end{corollary}

\begin{proof}
This follows from the fact that the $\NP$-hardness proof for slice rank \cite{BlaserILPS21} uses only instances whose support is an antichain. For these instances, symmetric slice rank and slice rank coincide by Corollary~\ref{cor:symSR-SR}.
In particular, by the same pivot argument as in Lemma~\ref{lem:symSR-tau}, the maximum dimensions of formally independent and
functionally independent spaces coincide, and the hardness proof
also works for $\FF_2$.
\end{proof}

\begin{corollary}
Computing $\Str(P)$ is $\NP$-hard.  
\end{corollary}

\begin{remark}\label{rem:TI-space-NP}
  A related result is the following. That is, for a bilinear map $\phi:U\times U\to V$ where $U\cong \FF^n$ and $V\cong \FF^m$, we say that $W\leq U$ is a totally-isotropic space of $\phi$ if for any $w, w'\in W$, $\phi(w, w')=0$. In \cite{bei_independent_2019}, it is shown that the maximum totally-isotropic space dimension of alternating bilinear maps is \np-hard, by a reduction from the independent set problem on graphs. Here, our proof of \np-hardness for the maximum independent space dimension for symmetric trilinear forms is more involved due to the lack of corresponding graph or hypergraph problems to start with.
\end{remark}

\section{Cryptanalysis of the planted independent space problem}
\label{sec:cryptanalysis}

\subsection{The $h/n > 1/2$ regime}
\label{sub:the_h_gtrsim_n_2_regime}

\subsubsection{Search algorithm for \(h\ge(1/2+\epsilon)n\) through rank profiles} \label{subsubsec:rank-profile}

Let $f$ be a cubic Boolean polynomial with constant term $0$. Recall that $\TF^f$ is the alternating trilinear form associated with $f$, as defined in \cref{def:Tf}.

In the generation phase of the planted distribution $\mc P_3(n, h)$, an $h$-dimensional planted independent space is constructed by forcing an $h \times h \times h$ corner block of the tensor to be entirely zero, followed by a secret random change of basis to hide this structure. In other words, there exists a secret $h$-dimensional subspace $V \le \mathbb{F}_2^n$ such that $\TF^f|_{V^3} = 0$. 

By \cref{def:Tf} we know that $\TF^f$ captures the homogeneous degree-3 part of $f$:
$\TF^f=\TF^{f^{(3)}}$. For the purpose of our algorithm, we only deal with this homogeneous cubic component, which is a cubic form. In the remainder of this subsection, we assume that $f$ is homogeneous, that is, a cubic form.
Since $f$ has no quadratic part, \cref{cor:equiv} gives $\TF^f_v=\Hess^f_v$ for every $v\in\mathbb{F}_2^n$. Our algorithm exploits the rank profile of these tensor slices, which is invariant under invertible linear changes of variables.
We show that collecting sufficiently many linearly independent vectors whose slices have low rank allows us to recover the planted independent space.

\begin{algorithm}[H]
\caption{Rank-Profile Search Algorithm for Cubic Forms with
Planted Independent Spaces}
\label{alg:rank-profile-search}
\textbf{Input:} A homogeneous degree-3 Boolean polynomial
$f:\mathbb{F}_2^n\to\mathbb{F}_2$, and
\(h\ge(1/2+\epsilon)n\). \\
\textbf{Output:} A basis of the planted independent space,
or \texttt{Failure}.
\begin{enumerate}
  \item Let $N=4h\cdot 2^{n-h}$ and initialize an empty
  list $\mathcal B$.

  \item For $t=1$ to $N$:
  \begin{enumerate}
      \item Sample a vector
      $v^{(t)}\in\mathbb{F}_2^n$ uniformly at random,
      independently of all previous samples.

      \item Compute the Hessian slice matrix
      $M_{v^{(t)}}\in\mathbb{F}_2^{n\times n}$, where
      $(M_{v^{(t)}})_{ij}
      =\partial_j\partial_i f(v^{(t)})$.

      \item Compute the rank of $M_{v^{(t)}}$ over
      $\mathbb{F}_2$ using Gaussian elimination.

      \item If
      $\operatorname{rank}(M_{v^{(t)}})\le 2(n-h)$
      and $v^{(t)}\notin\operatorname{span}(\mathcal B)$,
      append $v^{(t)}$ to $\mathcal B$.

      \item If $|\mathcal B|=h$, output $\mathcal B$
      and halt.
  \end{enumerate}

  \item If the loop terminates without halting,
  output \texttt{Failure}.
\end{enumerate}
\end{algorithm}

\begin{theorem}
\label{thm:rank-profile-search}
  For any fixed constant \(0<\epsilon<1/2\) and sufficiently large \(n\), let $h$ be an integer satisfying $(1/2+\epsilon)n\le h\le n$, and let $f$ be the cubic component of a polynomial sampled from $\mc P_3(n,h)$. Algorithm~\ref{alg:rank-profile-search} recovers the planted
  independent space in time $O(hn^3 2^{n-h})$ and polynomial space with probability $1-e^{-\Omega(n)}$. 
\end{theorem}

\begin{proof}
Since $f$ is a homogeneous multilinear cubic form, \cref{cor:equiv} gives $M_v=\Hess^f_v=\TF^f_v$ for every $v\in\FF_2^n$. For the analysis, we regard all $N$ vectors as sampled independently in advance, even if the algorithm halts before examining all of them.

Suppose $f$ contains an $h$-dimensional independent space $V$, so that $\TF^f|_{V^3}=0$. Whenever $v\in V$, we have $\TF^f(v,x,y)=0$ for all $x,y\in V$. 
In a basis whose first $h$ vectors span $V$, the matrix $M_v$ therefore has the block form
$$M_{v^{(t)}} = \begin{pmatrix} 0_{h \times h} & B \\ B^T & D \end{pmatrix}.$$
Since $B$ has dimensions $h \times (n-h)$, the rank of the first $h$ columns is at most $n-h$, and the rank of other $n-h$ columns is at most $n-h$. Therefore we obtain:
\[
\text{rank}(M_{v^{(t)}}) \le 2(n-h) <n.
\]
Thus, every sampled vector in $V$ satisfies the rank condition in Step 2(d).

Next, we bound the probability that a sampled vector outside $V$ satisfies the same rank condition. By the planted sampling procedure, $\TF^f$ is uniformly distributed over the space of alternating trilinear forms satisfying $\TF^f|_{V^3}=0$.

For any fixed $v\notin V$, choose a basis $v,w_1,\ldots,w_{n-1}$ with $V=\operatorname{span}\{w_1,\ldots,w_h\}$. The planted condition only constrains the coefficients $\TF^f(w_i,w_j,w_k)$ with $1\le i<j<k\le h$. Thus, the coefficients $\TF^f(v,w_j,w_k)$,
$1\le j<k\le n-1$, remain independent uniform bits. Since $\TF^f(v,v,w)=0$ for every $w$, the matrix $M_v$ in this basis has the form $0\oplus A$, where $A$ is a uniformly random alternating $(n-1)\times(n-1)$ matrix.

For a uniformly random alternating $m\times m$ matrix $A$, requiring a fixed $d$-dimensional subspace to lie in $\ker A$ imposes $d(m-d)+\binom d2$ independent linear constraints. There are at most $4\cdot2^{d(m-d)}$ such subspaces, so a union bound gives
\[
    \Pr[\dim\ker A\ge d] \le 4\cdot2^{-\binom d2}, \qquad 0\le d\le m.
\]
Set $p=2^{h-n}$ and $r=2h-n-1$, which is positive for sufficiently large $n$. Applying this bound with $m=n-1$ and $d=r$, and taking a union bound over the $N$ samples, we obtain
\[
    \Pr\bigl[\exists t\in[N]:
        v^{(t)}\notin V,\rank(M_{v^{(t)}})\le2(n-h)
    \bigr] \le 4N(1-p)\,2^{-\binom r2}.
\]

It remains to show that the sampled vectors in $V$ span $V$. If they do not, they are all contained in some hyperplane $H<V$. There are $2^h-1$ such hyperplanes, and for each one,
$|V\setminus H|=2^{h-1}$. Each independent sample therefore avoids $V\setminus H$ with probability $1-p/2$. A union bound gives
\[
    \Pr\bigl[\operatorname{span}\{v^{(t)}:t\in[N],\ v^{(t)}\in V\}
        \neq V\bigr]\le (2^h-1)\left(1-\frac{p}{2}\right)^N\le 2^h e^{-\frac{Np}{2}}=e^{-(2-\ln2)h},
\]
where we used $N=4h/p$.

Unless one of these two events occurs, every vector appended to $\mathcal B$ lies in $V$, and the sampled vectors in $V$ span $V$. Hence the algorithm returns a basis of $V$, with failure probability at most
\[
    4N(1-p)\,2^{-\binom r2}+e^{-(2-\ln2)h}\le 16h\,2^{n-h-\binom r2}+e^{-(2-\ln2)h}.
\]
Since $r\ge2\epsilon n-1$, the first term is $2^{-\Omega(n^2)}$ and the second is
$e^{-\Omega(n)}$. This proves the claimed success probability. The bound holds for every fixed $V$ and therefore also after averaging over the choice of the planted space.

Finally, each iteration takes $O(n^3)$ operations, including constructing the slice, computing its rank, and updating a basis for $\operatorname{span}(\mathcal B)$.
Thus, the total running time is $O(hn^3 2^{n-h})$.    
\end{proof}

\begin{remark}
    The same technique also gives a distinguisher: using only $N= c\,2^{n-h}$ samples, where $c$ is a sufficiently large constant, output \texttt{Planted} as soon as one sampled vector satisfies the rank condition, and \texttt{Random} otherwise.
\end{remark}

\begin{remark}
    Our algorithm works when \(h\ge(1/2+\epsilon)n\).  
    A similar idea was used in \cite{liu_planted_2026} in the context of planting a totally-isotropic space in an alternating bilinear map (see \cref{rem:TI-space-NP}). There, a similar distinguisher runs in  polynomial time for $h>n/2$. For $h=n/2$ and even $h=n/2-o(n)$, the algorithm is also  in sub-exponential time. The main idea is to view the tensor as $m$ matrix slices $A_1, \dots, A_m$ and analyse their rank profiles. Any linear combination of $A_i$ always has an $h\times h$ zero block on the top-left so it is easy to distinguish by the probability of certain rank deficiency.
    
    While the distinguishers in \cite{liu_planted_2026} and the search algorithm in \cref{thm:rank-profile-search} follow the same strategy, the results are quantitatively different. Recall that in the process of sampling from $\mc P_3(n, h)$, we set the upper-left-frontal $h\times h\times h$ corner block to be $0$.
    The Hessian slice matrix $M_v$ exhibits a top-left zero block whenever $v$ belongs to the independent space, whose probability is $2^{h-n}$. So in step 2 of \Cref{alg:rank-profile-search}, we need to sample $N=2^{n-h}$ vectors to obtain one in the isotropic space with high probability, so that already demonstrates a major difference between our setting and the setting in \cite{liu_planted_2026}. Finally, we note that when $h>(1/2+\epsilon)n$, the rank helps to distinguish between whether $v\in H$ or not. When $h \le n/2$, we need to compute the probability of certain rank deficiency to distinguish between a random matrix and a matrix with zero block.
\end{remark}

\subsubsection{Search algorithm for $h\geq (1/2+\epsilon)n$ through computing shrunk subspaces} 
We now devise another $2^{n-h}\cdot \poly(n)$-time
algorithm to find a planted independent space of dimension $h\geq (1/2+\epsilon)n$ with high probability. While this algorithm may not be more efficient than the algorithm in \cref{thm:rank-profile-search}, it uses fewer samples from $\FF_2^n$ and introduces another technique (shrunk subspace computation), so we decide to record it here.

We basically follow the setting in \cref{subsubsec:rank-profile}. Let $f\in\FF_2[x_1, \dots, x_n]$ be a 
cubic form. Recall the alternating trilinear form $\TF^f$ associated with $f$ defined in \cref{def:Tf}.
Let $V\leq \FF_2^n$ be the planted independent space, satisfying that $h=\dim(V)=\lceil (1/2+\epsilon)n\rceil$. 

The algorithm follows the same structure as \cref{alg:rank-profile-search}, but it requires the following preparation. In the following, $C$ is a constant that will be determined later. 
\begin{definition}
  Let $\mathbf{A}=(A_1, \dots, A_C)\in \M(n, \FF)^C$ be a matrix tuple. For $U\leq\FF^n$, the image of $U$ under $\mathbf{A}$ is $\mathbf{A}(U)=\mathrm{span}\{\cup_{i\in[C]}A_i U\}$. 
  
  For $c\in\NN$, we say that $U$ is a $c$-shrunk subspace of $\mathbf{A}$, if $\dim(U)-\dim(\mathbf{A}(U))\geq c$. Equivalently, $U$ is a $c$-shrunk subspace of $\mathbf{A}$, if there exists $W\leq\FF^n$ of dimension $n-\dim(U)+c$, such that for any $u\in U$ and $w\in W$, and any $i\in[C]$, $w^tA_iu=0$.
\end{definition}

Let $T$ be a random $n\times n\times n$ alternating tensor over $\FF_2$, with the top-left-front $h\times h\times h$ subcube set to $0$.
Let $M_1, \dots, M_C$ be generic combinations of the first $h$ frontal slices of $T$. Suppose $n\gg C$. Note that $M_i=\begin{bmatrix}
    0 & M_i' \\
    -M_i'^t & M_i''
\end{bmatrix}$, where $0$ indicates the zero matrix of size $h\times h$. Furthermore, the entries in $M_i'$ and the upper triangular part of $M_i''$, for $i\in[C]$, are independent. That is, we may view $M_i$'s as alternating matrices whose upper-triangular entries are independent.

Let $H$ be the subspace spanned by the first $h$ standard basis vectors. Note that $H$ is a $c$-shrunk subspace where $c=\lceil (1/2+\epsilon)n\rceil -(n-\lceil (1/2+\epsilon)n\rceil)\approx 2\epsilon n$. Indeed, viewing $M_1, \dots, M_C$ as bilinear forms on $\FF_2^n$, $H\leq \FF_2^n$ satisfies that for any $i\in[C]$, $M_i(H, H)=0$. 
\begin{proposition}\label{prop:unique}
    Let $C\geq 4$, and let $M_1, \dots, M_C$ be random $n\times n$ alternating matrices over $\FF_2$ with the top-left $h\times h$ submatrix set to $0$, with $h\geq (1/2+\epsilon)n$. Then with high probability, $(V,W)=(H, H)$ is the only pair of subspaces of dimension $h$ satisfying that $\forall i\in[C]$, $M_i(V, W)=0$. Furthermore, $H$ is the $(2h-n)$-shrunk subspace of the smallest dimension.
\end{proposition}
\begin{proof}
Let $(V, W)$ be a pair of subspaces of dimension $h$. 

First, consider the case when $V=H$. Suppose $W\cap H$ is of codimension $k$ in $W$, $k\geq 1$. Choose $w_1, \dots, w_k\in W$ such that $(W\cap H)\cup \{w_1, \dots, w_k\}$ span $W$. Then $M_i(H, W)=0$ for every $i\in[C]$ implies that $M_i(H, w_j)=0$ for every $i\in[C]$ and $j\in [k]$, which has probability $\leq 1/2^{k\left(Ch-\binom{C+1}{2}\right)}$.
The number of $W$ such that $W\cap H$ is of codimension $k$ in $W$ is upper bounded by $2^{k(h-k)+kn}$. As $h\geq (1/2+\epsilon)n$, when $C\geq 6$ is a constant, we have $k(Ch-\binom{C+1}{2})\geq k(Ch-C^2)=kC(h-C)\geq 3nk-kC^2\geq hk+nk+(nk-kC^2)\geq k(h-k)+kn$; note that we assume $n\geq C^2$. Therefore, by the union bound, when $C\geq 4$, the probability of having $(H, W)$, $W\neq H$ such that $M_i(H, W)=0$ for $i\in[C]$ is negligible.

Second, the case of $V\neq H$ can be dealt with in the same approach, considering their intersections with $H$. We omit the routine calculations here. 

For the furthermore statement, note that the existence of a $(2h-n)$-shrunk subspace of dimension larger than $h$ is equivalent to the existence of $(V, W)$ where $\dim(V)=h+d$, $\dim(W)=h-d$, where $d\geq 1$, such that $M_i(V, W)=0$ for $i\in[C]$. Again, it can be shown that the probability of random $M_i$'s allowing such $(V, W)$ is negligible by specifying bases of $V$ and $W$ with respect to $H$. 
\end{proof}

Proposition~\ref{prop:unique} ensures that $H$ is the $c$-shrunk subspace where $c=2h-n$ of the smallest dimension. Such a shrunk subspace is called the canonical shrunk subspace in \cite{ivanyos2022symbolic}, which can be computed by the algorithm in \cite{ivanyos_constructive_2018}.

We can now describe the algorithm. As it is similar to \cref{alg:rank-profile-search}, we present a sketch instead of a formal description. 
\begin{enumerate}
\item Let $N=4C\cdot 2^{n-h}$. Initialize an empty list $\mathcal{B}$. 
\item For $t=1\to N$:
\begin{enumerate}
    \item Sample a random $v^{(t)}\in \FF_2^n$. 
    \item Let $M_t$ be the Hessian slice matrix $\Hess^f_{v^{(t)}}$.
    \item Compute the rank of $M_t$ over $\FF_2$. 
    \item If
        $\operatorname{rank}(M_{t})\le 2(n-h)$
        and $v^{(t)}\notin\operatorname{span}(\mathcal B)$,
        append $v^{(t)}$ to $\mathcal B$.
    \item If $|\mathcal B|=C$, continue to the next step.
\end{enumerate}
\item Compute the canonical shrunk subspace $V'$ of matrices in $\mathcal{B}$ using \cite{ivanyos_constructive_2018}. 
\item Return $V'$ if $\dim(V')=\lceil (1/2+\epsilon)n\rceil$. Report ``Failure'' otherwise.
\end{enumerate}

Following the analysis of \cref{thm:rank-profile-search}, we see that the algorithm succeeds with high probability in sampling $C$ matrices in the planted space $H$ after step 2. Then by \cref{prop:unique}, the shrunk subspace $V'$ of matrices in $\mathcal{B}$ coincides with the planted subspace $H$. 

Clearly, the algorithm runs in time $2^{n-h}\cdot\poly(n)$, and uses $O(C\cdot 2^{n-h})$ many samples from $\FF_2^n$, instead of $O(h\cdot 2^{n-h})$-many in \cref{alg:rank-profile-search}.

\subsection{Analysis of Bell-sampling attacks}

In this section, we consider a potential attack on the security of \decision-$h$-\PIS\ based on Bell sampling. 
The motivation is the following: it has been shown that a certain measure of magic for quantum states---the \emph{stabilizer nullity} \cite{beverland_lower_2020}---can be efficiently estimated using Bell sampling~\cite{hangleiter_bell_2024,grewal_improved_2024}. 
On the other hand, Bell sampling can be efficiently (strongly) simulated for degree-$3$ circuits, since every outcome amplitude is described by a Clifford circuit \cite{hangleiter_fault-tolerant_2025}. 
But the entire point of degree-$3$ circuits with a planted independent space is that they have low stabilizer rank compared to random states ($2^{cn}$ compared to $2^{n}$). 
Therefore, one may be worried that Bell sampling can be used to efficiently solve \decision-$h$-\PIS\ for arbitrary~$h$. 

The resolution of why this attack does not work is that while the stabilizer nullity gives an upper bound to the stabilizer rank this upper bound is ``maximally'' loose for vertex-cover decompositions of random circuits. 

\begin{definition}[Stabilizer nullity]
  Given a quantum state $\ket \psi \in (\mb C^{2})^{\otimes n}$, the stabilizer nullity is defined as 
  \begin{align}
    \nu(\ket{\psi})  = n - \dim(\mc S). 
  \end{align}
  Here, $\mc S \le \mc P_n$ is the (Pauli) stabilizer group of $\ket \psi$, i.e., $S \in \mc S \Rightarrow S \ket \psi = \ket \psi$. 
\end{definition}

\begin{lemma}
Let  $h \le n-2 $ and $g$ be a random cubic polynomial with planted $h$-dimensional independent space. Then 
\begin{align}
  \Pr_{g \sim \mc P_3(n,h)}[\nu(\ket g) < n ]\le 2^{-n+4}. 
  \end{align}
\end{lemma}
\begin{proof}
Consider the stabilizer group
$\mc S(\ket{+^n})=\langle X_i\rangle_{i\in[n]}$.
Then a stabilizer subgroup of $\ket g$ is generated
by the Clifford operators $S_i= C(g) X_i C(g)^\dagger = X_i\cdot C(g_i)$ where
\begin{align}
  g_i(x) = g(x+e_i)+g(x)=
    \sum_{j < k} G_{ijk} x_jx_k + 
    \ell_i(x),
\end{align}
where $G_{ijk}=\TF^g(e_i,e_j,e_k)$ and
$\ell_i$ is affine linear.

The stabilizer nullity is determined by the number of linearly
independent \emph{Pauli stabilizers}.
We therefore need to determine how many Pauli stabilizers
the group $\langle S_i\rangle_i$ contains.
To this end, we observe that
\begin{align}  
  S_i S_j = X_i X_j C(g_i(\,\cdot + e_j) +  g_j),
\end{align}
and therefore an element
$S(v)\in \langle S_i \rangle_i$
labelled by $v\in\FF_2^n$ is given by
\begin{align}
  S(v)=X(v)C(g_v),\qquad
  g_v(x) = g(x+v)+g(x).
\end{align}
Every Pauli stabilizer belongs to this group: the diagonal part $C(g_v)$ is uniquely determined by its $X$-part and the nonzero computational-basis amplitudes of $\ket g$.
$S(v) \in \mc P_n$ iff $g_v$ is affine linear, that is, whenever
$G(v)=\sum_i v_i G_{ijk} \in \FF_2^{n\times n}$ is diagonal.
Letting $G'_{i,jk}=G_{ijk}$ for $j<k$,
$i\notin\{j,k\}$ and $0$ otherwise, this is the case iff
$v^TG'=0$, and therefore the stabilizer nullity of~$\ket g$
is given by
\begin{align}
  \nu(\ket g)=\rank(G').
\end{align}

Next, we show that $G'$ has full rank with high probability. Wlog.\ let the independent space of $g$ be spanned by $e_1, \ldots, e_h$. 
Given $v \neq 0$, we bound $\Pr[v^T G' = 0]$ by exhibiting many entries of the row vector $v^T G'$ that are independent uniform bits.
Pick any $i_0$ with $v_{i_0} = 1$ and consider only the entries of $v^T G'$ indexed by pairs $\{j,k\}$ with $i_0 \notin \{j,k\}$ and $\{i_0,j,k\} \not\subseteq [h]$. Each such entry has the form
\begin{align}
  (v^T G')_{jk} = G_{i_0jk} + \sum_{\substack{i \in \supp(v)\setminus\{i_0\} \\ i \notin \{j,k\}}} G_{ijk} ,
\end{align}
and the coefficient $G_{i_0jk}$ appears in no other entry under consideration.\footnote{Its only other occurrences in $v^T G'$ are in the entries $i_0j$ and $i_0k$, which we have excluded. }
Conditioned on all remaining coefficients, these entries are therefore independent uniformly random  bits so that 
\begin{align}
  \Pr[v^T G' = 0]  \le 2^{-N(v)}, \qquad N(v) = \left|\bigl\{\{j,k\} \subseteq [n]\setminus\{i_0\} : \{i_0,j,k\} \not\subseteq [h]\bigr\}\right|.
\end{align}
\begin{itemize}
  \item If $\supp(v) \not\subseteq [h]$, choose $i_0 > h$  and  $\{i_0,j,k\} \not\subseteq [h]$  holds automatically, giving  $N(v) = \binom{n-1}{2}$.
  \item If $\supp(v) \subseteq [h]$, then $i_0 \le h$ and $N(v) = \binom{n-1}{2} - \binom{h-1}{2} = \frac{(n-h)(n+h-3)}{2}$.
\end{itemize}
The union bound over the at most $2^n$ (respectively at most $2^h$) vectors of each type gives
\begin{align}
  \Pr[\rank(G') < n] \le 2^{\,n - \binom{n-1}{2}} + 2^{\,h - \frac{(n-h)(n+h-3)}{2}} .
\end{align}
The exponent of the second term is increasing in $h$ and hence maximized at $h = n-2$, where it equals $-(n-3)$; the first term is at most $2^{-(n-3)}$ for $n \ge 6$. This gives the claimed $2^{-n+4}$. 
\end{proof}

Note that for $h \le (1-\epsilon)n$ the same argument gives the much stronger bound $2^{-\Omega_\epsilon(n^2)}$.

\printbibliography

\appendix 

\section{Attack on block-diagonal structure}
\label{app:attack_on_block_diagonal_structure}

Suppose we have tensors $T_i\in U_i\otimes V_i\otimes W_i$, $i\in[s]$. Let $U=\oplus_i U_i$, $V=\oplus_i V_i$, and $W=\oplus_i W_i$. We can then form the direct sum of $T_i$ as $\oplus_i T_i\in U\otimes V\otimes W= (\oplus_i U_i)\otimes (\oplus_i V_i)\otimes (\oplus_i W_i)$. Such a block-diagonal tensor $\oplus_i T_i$ can then be transformed to $T$ via $\gl(U)\times\gl(V)\times\gl(W)$. Therefore, $T$ is isomorphic to a block-diagonal tensor $\oplus_i T_i$, and our task is to compute the transformation that reveals this block-diagonal structure. 

This can be done via some known procedures in computer algebra. In the following, by slight abuse of notation, we assume that $T$ is a trilinear form $T:U\times V\times W\to\FF$. For $u\in U$, let $T_u:V\times W\to \FF$ be the bilinear form by specifying the first argument of $T$ to $u$. Similarly, for $v\in V$ and $w\in W$, we can define $T_v$ and $T_w$ analogously.

First, we may assume that $T$ (and therefore $\oplus_i T_i$) is non-degenerate (or concise by Strassen). Recall that $T$ is degenerate, if there exists $x=u\in U$, or $x=v\in V$, or $x=w\in W$, such that $T_x$ is the zero form.

\newcommand{\End}{\mathrm{End}}
\newcommand{\Cent}{\mathrm{Cent}}
Second, for a non-degenerate $T$, its centroid algebra $\Cent(T)$ consists of $(A, B, C)\in\End(U)\oplus\End(V)\oplus\End(W)$, such that for any $u\in U$, $v\in V$, and $w\in W$, $T(Au, v, w)=T(u, Bv, w)=T(u, v, Cw)$. Here, $\End(U)$ denotes the endomorphism algebra of $U$. Note that a linear basis of $\Cent(T)$ can be computed efficiently. It can be shown that $\Cent(T)$ is a commutative algebra, and $T$ is isomorphic to a block-diagonal tensor if and only if $\Cent(T)$ contains a non-trivial idempotent, that is, $(P, Q, R)\in \Cent(T)$ such that $P^2=P$, $Q^2=Q$, and $R^2=R$. 

Third, we are reduced to the question of deciding whether $\Cent(T)$ contains a non-trivial idempotent. This can be solved by first computing the algebra structure of $\Cent(T)$, namely its Jacobson radical and the semisimple quotient decomposition. The non-trivial idempotent task is then solved by first working this problem out in the simple components, and if every simple component admits a non-trivial idempotent, gluing the solutions and lifting through the radical.

\end{document}